\documentclass[11pt]{article}

\usepackage{fullpage}

\usepackage[utf8]{inputenc} \usepackage[T1]{fontenc}    \usepackage{hyperref}       \usepackage{url}            \usepackage{booktabs}       \usepackage{amsfonts}       \usepackage{nicefrac}       \usepackage{microtype}      
\usepackage{amsmath,amsthm,amssymb,epsfig,color,float,verbatim}
\usepackage{enumerate}
\usepackage[noend]{algorithmic}
\usepackage{algorithm}

\newtheorem{theorem}{Theorem}

\newtheorem{lemma}[theorem]{Lemma}
\newtheorem{proposition}[theorem]{Proposition}
\newtheorem{corollary}[theorem]{Corollary}
\newtheorem{claim}[theorem]{Claim}
\newtheorem{fact}[theorem]{Fact}

\newtheorem{definition}[theorem]{Definition}
\newtheorem{question}{Question}

\usepackage{ifthen}
\newboolean{hasAppendix}
\setboolean{hasAppendix}{true}

\newcommand{\wt}[1]{{\widetilde{#1}}}

\newcommand{\R}{\mathbb{R}}

\newcommand{\E}{\mathbb{E}}

\newcommand{\poly}{\mathrm{poly}}

\newcommand{\dtv}{d_{\mathrm TV}}

\newcommand{\dhel}{d_{\mathrm H}}

\newcommand{\Cov}{\mathrm{Cov}}

\newcommand{\diag}{\mathrm{diag}}

\usepackage{enumitem}

\newcommand{\ignore}[1]{}

\newcommand{\eps}{\epsilon}

\usepackage{mathtools}
\DeclarePairedDelimiter{\ceil}{\lceil}{\rceil}
\newcommand{\CondTable}{\Gamma}

\usepackage{soul}
\newcommand{\yucomment}[1]{{}}

\newcommand{\eqdef}{\stackrel{{\mathrm {\scriptstyle def}}}{=}}

\title{Robust Learning of Fixed-Structure Bayesian Networks}
\author{
  Yu Cheng\\
  Duke University\\
  \texttt{yucheng@cs.duke.edu}\\
\and
  Ilias Diakonikolas\\
  University of Southern California\\
  \texttt{ilias.diakonikolas@gmail.com}\\
\and
  Daniel M. Kane\\
  University of California, San Diego\\
  \texttt{dakane@ucsd.edu}
\and
  Alistair Stewart\\
  University of Southern California\\
  \texttt{stewart.al@gmail.com}\\
}

\begin{document}

\maketitle

\begin{abstract}
We investigate the problem of learning Bayesian networks in a robust model
where an $\eps$-fraction of the samples are adversarially corrupted.
In this work, we study the fully observable
discrete case where the structure of the network is given.
Even in this basic setting, previous learning algorithms
either run in exponential time or lose dimension-dependent factors in their
error guarantees.
We provide the first computationally efficient robust learning algorithm for this problem
with dimension-independent error guarantees. Our algorithm has near-optimal sample complexity,
runs in polynomial time, and achieves error that scales nearly-linearly with the fraction
of adversarially corrupted samples. Finally, we show on both synthetic and semi-synthetic 
data that our algorithm performs well in practice.
\end{abstract}

\makeatletter{}
\section{Introduction}  \label{sec:intro}

Probabilistic graphical models~\cite{Koller:2009} provide an appealing and unifying formalism to succinctly
represent structured high-dimensional distributions. The general problem of inference in
graphical models is of fundamental importance and arises in many applications across several scientific disciplines,
see~\cite{Wainwright:2008} and references therein.
In this work, we study the problem of learning graphical models from data~\cite{Neapolitan:2003, RQS11}.
There are several variants of this general learning problem depending
on: (i) the precise family of graphical models considered (e.g., directed, undirected),
(ii) whether the data is fully or partially observable, and (iii) whether the structure of the underlying
graph is known a priori or not (parameter estimation versus structure learning).
This learning problem has been studied extensively along these axes
during the past five decades, (see, e.g.,~\cite{Chow68, Dasgupta97, Abbeel:2006, WainwrightRL06, AnandkumarHHK12,
SanthanamW12, LohW12, BreslerMS13, BreslerGS14a, Bresler15}) 
resulting in a beautiful theory and a collection of algorithms in various settings.

The main vulnerability of all these algorithmic techniques is that they crucially rely on the assumption
that the samples are precisely generated by a graphical model in the given family. This simplifying
assumption is inherent for known guarantees in the following sense: if there exists even a very small fraction
of arbitrary outliers in the dataset, the performance of known algorithms can be totally compromised.
It is important to explore the natural setting when the aforementioned assumption holds only in an approximate sense.
Specifically, we study the following broad question:

\begin{question}[Robust Learning of Graphical Models]
\label{question}
Can we efficiently learn graphical models when a constant fraction of the samples are corrupted, or equivalently, when the model is slightly misspecified?
\end{question}

In this paper, we focus on the model of corruptions considered in~\cite{DiakonikolasKKL16} (Definition~\ref{def:adv})
which generalizes many other existing models, including Huber's contamination model~\cite{Huber64}.
Intuitively, given a set of good samples (from the true model), an 
adversary is allowed to inspect the samples before corrupting them,
both by adding corrupted points and deleting good samples. In contrast, in Huber's model,
the adversary is oblivious to the samples and is only allowed to add bad points.

We would like to design robust learning algorithms for Question~\ref{question} whose
sample complexity, $N$, is close to the information-theoretic minimum,
and whose computational complexity is polynomial in $N$. We emphasize that the crucial requirement
is that the error guarantee of the algorithm is {\em independent} of the dimensionality $d$ of the problem.

\subsection{Formal Setting and Our Results}  \label{ssec:results}

In this work, we study Question~\ref{question} in the context of {\em Bayesian networks}~\cite{Jensen:2007}.
We focus on the fully observable case when the underlying network is given.
In the {\em non-robust} setting, this learning problem is straightforward:
the ``empirical estimator'' (which coincides with the maximum likelihood estimator)
is known to be sample and computationally efficient~\cite{Dasgupta97}.
In sharp contrast, even this most basic regime is surprisingly challenging in the robust setting.
For example, the very special case of robustly learning a
Bernoulli product distribution (corresponding to an empty network with no edges)
was analyzed only recently in~\cite{DiakonikolasKKL16}.

To formally state our results, we first give a detailed description of the corruption model we study. \begin{definition}[$\eps$-Corrupted Samples]
\label{def:adv}
Given $0 < \eps < 1/2$ and a distribution family $\mathcal{P}$, 
the algorithm specifies some number of samples $N$, and $N$ samples $X_1, X_2, \ldots, X_N$ are drawn from some (unknown) ground-truth $P \in \mathcal{P}$.
The adversary is allowed to inspect $P$ and the samples, and replaces $\eps N$ of them with arbitrary points. The set of $N$ points is then given to the algorithm.
We say that a set of samples is $\eps$-corrupted if it is generated by this process. \end{definition}

\paragraph{Bayesian Networks.}
Fix a directed acyclic graph, $G$, whose vertices are labelled $[d] \eqdef \{1,2,\ldots,d\}$
in topological order (every edge points from a vertex with a smaller index to one with a larger index).
We will denote by $\text{Parents}(i)$ the set of parents of node $i$ in $G$.
A probability distribution $P$ on $\{0,1\}^d$ is defined to be 
a \emph{Bayesian network} (or \emph{Bayes net}) with graph $G$
if for each $i \in [d]$, we have that $\Pr_{X\sim P}\left[X_i = 1 \mid X_1,\ldots,X_{i-1}\right]$ 
depends only on the values $X_j$ where $j \in \text{Parents}(i)$.
Such a distribution $P$ can be specified by
its {\em conditional probability table}. 
\begin{definition}[Conditional Probability Table of Bayesian Networks]
Let $P$ be a Bayesian network with graph $G$.
Let $\CondTable$ be the set $\{(i,a): i \in [d], a\in \{0,1\}^{|\mathrm{Parents}(i)|}\}$.
Let $m=|\CondTable|$. For $(i,a)\in \CondTable$, the \emph{parental configuration} $\Pi_{i,a}$
is defined to be the event that $X_{\mathrm{Parents}(i)}=a$.
The conditional probability table $p\in [0,1]^m$ of $P$ is given by
$p_{i,a}=\Pr_{X\sim P}\left[X_i=1 \mid \Pi_{i,a}\right].$
\end{definition}
Note that $P$ is determined by $G$ and $p$.
We will frequently index $p$ as a vector.
We use the notation $p_k$ and the associated events $\Pi_k$, where each $k \in [m]$ stands for an $(i,a) \in \CondTable$ lexicographically ordered.

\paragraph{Our Results.}
We give the first efficient robust learning algorithm for Bayesian networks with a known graph $G$.
Our algorithm has information-theoretically near-optimal sample complexity, runs in time polynomial in the size of the input (the samples),
and provides an error guarantee that scales near-linearly with the fraction
of adversarially corrupted samples, under the following restrictions: 
First, we assume that each parental configuration is reasonably likely. 
Intuitively, this assumption seems necessary
because we need to observe each configuration many times in order to learn
the associated conditional probability to good accuracy. Second, we assume
that each of the conditional probabilities is balanced, i.e., bounded away from $0$ and $1$.
This assumption is needed for technical reasons. In particular, we need this to show that
a good approximation to the conditional probability table implies that the corresponding Bayesian network is close in total variation distance.

Formally, we say that a Bayesian network is \emph{$c$-balanced}, for some $c>0$,
if all coordinates of the corresponding conditional probability
table are between $c$ and $1-c$. Throughout the paper, we use $m = \sum_{i=1}^d 2^{|\mathrm{Parents}(i)|}$ for the size of the conditional probability table of $P$, and $\alpha$ for the minimum probability of parental configuration of $P$: $\alpha = \min_{(i, a) \in S} \Pr_P[\Pi_{i,a}]$.
We now state our main result.

\begin{theorem}[Main] \label{thm:bal-bn}
Fix $0 < \eps < 1/2$. Let $P$ be a $c$-balanced Bayesian network on $\{0, 1\}^d$ with known structure $G$.
Assume $\alpha \ge \Omega(\eps\sqrt{\log(1/\eps)} / c)$.
Let $S$ be an $\eps$-corrupted set of 
$N = \wt{\Omega}(m\log(1/\tau)/\eps^2)$ samples from $P$.~\footnote{Throughout the paper, we use $\wt O(f)$ to denote $O(f \, \poly\!\log(f))$.}
Given $G, \eps, \tau$, and $S$,
we can compute a Bayesian network $Q$ such that, with probability at least $1-\tau$, 
$\dtv(P,Q) \leq \eps\sqrt{\ln(1/\eps)}/ (\alpha c)$.
Our algorithm runs in time $\wt O(N d^2 / \eps)$.
\end{theorem}

Our algorithm is given in Section~\ref{sec:alg}.
We first note that the sample complexity of our algorithm is near-optimal for learning Bayesian networks with known structure.
The following sample complexity lower bound holds even without corrupted samples:

\begin{fact}[Sample Complexity Lower Bound, \cite{CanonneDKS17}]
\label{fact:sample-lb}
Let $\mathcal{BN}_{d,f}$ denote the family of Bernoulli Bayesian networks on $d$ variables such that every node has at most $f$ parents.
The worst-case sample complexity of learning $\mathcal{BN}_{d,f}$, within total variation distance $\eps$ and with probability $9/10$, is $\Omega(2^f \cdot d / \eps^2)$ for all $f \le d/2$ when the graph structure is known.
\end{fact}

Consider Bayes nets whose average in-degree is close to the maximum in-degree, that is, when $m = \Theta(2^f d)$, the sample complexity lower bound in Fact~\ref{fact:sample-lb} becomes $\Omega(m/\eps^2)$, so our sample complexity is optimal up to polylogarithmic factors.

We remark that Theorem~\ref{thm:bal-bn} is most useful when $c$ is a constant and the Bayesian network has bounded fan-in $f$. In this case, the condition on $\alpha$ follows from the $c$-balanced assumption:
When both $c$ and $f$ are constants, $\alpha = \min_{(i, a) \in S} \Pr_P[\Pi_{i,a}] \geq c^f$ is also a constant, so the condition $c^f \geq \Omega(\eps \sqrt{\log(1/\eps)})$ automatically hold when $\eps$ is smaller than some constant.
On the other hand, the problem of learning Bayesian networks is less interesting when the fan-in is too large.
For example, if some node has $f = \omega(\log(d))$ parents, then the size of the conditional probability table is at least $2^f$, which is super-polynomial in the dimension $d$.

\paragraph{Experiments.}
We performed an experimental evaluation of our algorithm on both synthetic and real data.
Our evaluation allowed us to verify the accuracy and the sample complexity rates
of our theoretical results. In all cases, the experiments validate the usefulness of our algorithm, which significantly outperforms previous approaches, 
almost exactly matching the best rate without noise.

\paragraph{Related Work.}
Question~\ref{question} fits in the framework of robust statistics \cite{Huber09, HampelEtalBook86}.
Classical estimators from this field can be classified into two categories: either  (i) they are computationally
efficient but incur an error that scales {\em polynomially} with the dimension $d$, or (ii)
they are provably robust (in the aforementioned sense) but are hard to compute.
In particular, essentially all known estimators in robust statistics (e.g., the Tukey median~\cite{Tukey75})
have been shown~\cite{JP:78, Bernholt, HardtM13} to be intractable in the high-dimensional setting.
\yucomment{Why are the following discussion on sample complexity meaningful?} We note that the robustness requirement does not typically pose
information-theoretic impediments for the learning problem. In most cases of interest
(see, e.g., \cite{CGR15, CGR15b, DiakonikolasKKL16}), the sample complexity
of robust learning is comparable to its (easier) non-robust variant.
The challenge is to design {\em computationally efficient} algorithms.

Efficient robust estimators are known for various
{\em low-dimensional} structured distributions (see, e.g.,~\cite{DDS14-toc, CDSS13, CDSS14, CDSS14b, AcharyaDLS16, ADLS17, DLS18}).
However, the robust learning problem becomes surprisingly challenging in high dimensions.
Recently, there has been algorithmic progress on this front:
\cite{DiakonikolasKKL16, LaiRV16}~give
polynomial-time algorithms with improved error guarantees for certain ``simple'' high-dimensional structured distributions.
The results of~\cite{DiakonikolasKKL16} apply to simple distributions, including Bernoulli product distributions, Gaussians, and mixtures thereof (under some natural restrictions).
Since the works of~\cite{DiakonikolasKKL16, LaiRV16}, computationally efficient robust estimation in high dimensions has received considerable attention 
(see, e.g.,~\cite{DKS17-sq, DiakonikolasKKLMS17, BalakrishnanDLS17, DiakonikolasKKLMS18, DiakonikolasKS18-mixtures, DiakonikolasKS18-nasty, HopkinsL18, KothariSS18, PrasadSBR2018, DiakonikolasKKLSS2018sever, KlivansKM18, DKS18-lr, LSLC18-sparse}).

\subsection{Overview of Algorithmic Techniques}
\label{sec:techniques}
Our algorithmic approach builds on the framework of~\cite{DiakonikolasKKL16} with new technical and conceptual ideas.
At a high level, our algorithm works as follows: 
We draw an $\eps$-corrupted set of samples from a Bayesian network $P$ with known structure, 
  and then iteratively remove samples until we can return the empirical conditional probability table. 

First, we associate a vector $F(X)$ to each sample $X$ so that learning the mean of $F(X)$
  to good accuracy is sufficient to recover the distribution.
In the case of binary products, $F(X)$ is simply $X$, while in our case we need to take into account additional information 
  about conditional means.

From this point, our algorithm will try to do one of two things:
  Either we show that the sample mean of $F(X)$ is close to the conditional mean
  of the true distribution (in which case we can already learn the ground-truth Bayes net $P$), or we are able to produce a {\em filter}, i.e.,   we can remove some of our samples, and it is guaranteed that we throw away more bad samples than good ones. 
If we produce a filter, we then iterate on those samples that pass the filter. 
To produce a filter, 
  we compute a matrix $M$ which is roughly the empirical covariance matrix of $F(X)$.
We show that if the corruptions are sufficient to notably disrupt the sample mean of $F(X)$,
  there must be many erroneous samples that are all far from the mean in roughly the same direction,
  and we can detect this direction by looking at the largest eigenvector of $M$.
If we project all samples onto this direction,
  concentration bounds of $F(X)$ will imply that almost all samples far from the mean are erroneous,
  and thus filtering them out will provide a cleaner set of samples.

\paragraph{Organization.}
Section \ref{sec:prelims} contains some technical results specific to Bayesian networks
  that we need. Section \ref{sec:alg} gives the details of our algorithm and an overview of its analysis. In Section~\ref{sec:experiment}, we present the experimental evaluations.
In Section~\ref{sec:concl}, we conclude and propose directions for future work.
\ifthenelse{\boolean{hasAppendix}}{}{Due to space constraints, we defer the proofs of the technical lemmas to the full version of the paper.}

\section{Technical Preliminaries} \label{sec:prelims}

The structure of this section is as follows:
First, we bound the total variation distance between two Bayes nets in terms of their conditional probability tables.
Second, we define a function $F(x,q)$, which takes a sample $x$ 
and returns an $m$-dimensional vector that contains information about the conditional means. Finally, we derive a concentration bound from Azuma's inequality.
\ifthenelse{\boolean{hasAppendix}}{Proofs from this section have been deferred to Appendix~\ref{app:prelims}.}{}

\begin{lemma} \label{cor:bal-min}
Suppose that: (i) $\min_{k \in [m]} \Pr_P[\Pi_k] \geq \eps$, and (ii) $P$ or $Q$ is $c$-balanced,
and (iii) $\frac{3}{c}\sqrt{\sum_k \Pr_P[\Pi_k] (p_k-q_k)^2} \leq \eps$.
Then we have that $\dtv(P,Q) \leq \eps.$
\end{lemma}

Lemma~\ref{cor:bal-min} says that to learn a balanced fixed-structure Bayesian network, it is sufficient to learn all the relevant conditional means.
However, each sample $x \sim P$ gives us information about $p_{i,a}$ only if $x \in \Pi_{i,a}$.
To resolve this, we map each sample $x$ to an $m$-dimensional vector $F(x,q)$, and ``fill in'' the entries that correspond to conditional means for which the condition failed to happen.
We will set these coordinates to their empirical conditional means $q$:

\begin{definition}
\label{def:Fxq}
Let $F(x,q)$ for $\{0,1\}^d \times \R^m \to \R^m$ be defined as follows:
If $x \in \Pi_{i,a}$, then $F(x,q)_{i,a}=x_i$, otherwise $F(x,q)_{i,a}=q_{i,a}$.
\end{definition}

When $q=p$ (the true conditional means), the expectation of the $(i,a)$-th coordinate of $F(X,p)$,
for $X \sim P$, is the same conditioned on either $\Pi_{i,a}$ or $\neg \Pi_{i,a}$.
Using the conditional independence properties of Bayesian networks,
we will show that the covariance of $F(x,p)$ is diagonal.

\begin{lemma} \label{lem:F-moments}
For $X \sim P$, we have $\E(F(X,p)) = p$.
The covariance matrix of $F(X,p)$ satisfies $\Cov[F(X,p)]= \diag(\Pr_P[\Pi_k]p_k(1-p_k))$.
\end{lemma}

Our algorithm makes crucial use of Lemma~\ref{lem:F-moments} (in particular, that $\Cov[F(X,p)]$ is diagonal)
  to detect whether or not the empirical conditional probability
  table of the noisy distribution is close to the conditional probabilities.

Finally, we will need a suitable concentration inequality that works under conditional independence properties.
We can use Azuma's inequality to show that the projections of $F(X,q)$ on any direction $v$ is concentrated around the projection of the sample mean $q$.
\begin{lemma} \label{lem:Azuma}
For $X \sim P$, any unit vector $v \in \R^d$, and any $q \in [0, 1]^m$, we have
\[
\Pr[|v \cdot (F(X, q) - q)| \geq T + \|p-q\|_2] \leq 2\exp(-T^2/2) \;.
\]
\end{lemma}

\makeatletter{}
\section{Robust Learning Algorithm} \label{sec:alg}

We first look into the major ingredients required for our filtering algorithm,   and compare our proof with that for product distributions in \cite{DiakonikolasKKL16} on a   more technical level.

In Section~\ref{sec:prelims}, we mapped each sample $X$ to   $F(X,q)$ which contains information about the \emph{conditional} means $q$,
  and we showed that it is sufficient to learn the mean of $F(X,q)$ to learn the ground-truth Bayes net.

Let $M$ denote the empirical covariance matrix of $(F(X,q)-q)$.
We decompose $M$ into three parts:
One coming from the ground-truth distribution, one coming from the subtractive error
(because the adversary can remove $\eps N$ good samples),
and one coming from the additive error (because the adversary can add $\eps N$ bad samples).
We will make use of the following observations:
\begin{enumerate}
\item[(1)] The noise-free distribution has a diagonal covariance matrix.
\item[(2)] The term coming from the subtractive error has no large eigenvalues.
\end{enumerate}
These two observations imply that any large eigenvalues of $M$
  are due to the additive error. Finally, we will reuse our concentration bounds
  to show that if the additive errors are frequently far from the mean in a known direction,
  then they can be reliably distinguished from good samples.

For the case of binary product distributions in~\cite{DiakonikolasKKL16}, (1) is trivial because the coordinates are independent;
  but for Bayesian networks we need to expand the dimension of the samples and fill in the missing entries properly. Condition (2) is due to concentration bounds, and for product distributions it follows from standard Chernoff bounds,
  while for Bayes nets, we must instead rely on martingale arguments and Azuma's inequality. 
\ifthenelse{\boolean{hasAppendix}}{
The main difference between the proof of correctness of our algorithm 
  and those given in \cite{DiakonikolasKKL16} lies in analyzing the mean, 
  covariance and tail bounds of $F(X,q)$, and showing that 
  its mean and covariance are well-behaved when $q$ is close to $p$ (see Lemma~\ref{lem:good-set} in Appendix~\ref{app:alg}).}{}

\subsection{Main Technical Lemma and Proof of Theorem~\ref{thm:bal-bn}}

First, we need to show that a large enough set of samples with no noise satisfy properties we expect 
from a representative set of samples. 
We need that the mean, covariance, and tail bounds of $F(X,p)$ behave like we would expect them to.
\ifthenelse{\boolean{hasAppendix}}{
This happens with high probability. 
The details are given in Lemma~\ref{lem:good-set} in Appendix~\ref{app:good}.
}{}
We call a set of samples that satisfies these properties $\eps$-good for $P$.

Our algorithm takes as input an $\eps$-corrupted multiset $S'$ of $N = \wt \Omega(m \log (1/\tau) / \eps^2)$ samples. We write $S'=(S \setminus L) \cup E$, where $S$ is the set of samples before corruption,
  $L$ contains the good samples that have been removed or (in later iterations) incorrectly rejected by filters,
  and $E$ represents the remaining corrupted samples.
We assume that $S$ is $\eps$-good.
In the beginning, we have $|E|+|L| \leq 2 \eps |S|$.
As we add filters in each iteration, $E$ gets smaller and $L$ gets larger. 
However, we will prove that our filter rejects more samples from $E$ than $S$, so $|E|+|L|$ must get smaller.

We will prove Theorem~\ref{thm:bal-bn} by iteratively running the following efficient filtering procedure: 
\begin{proposition}[Filtering] \label{prop:bal}
Let $0 < \eps < 1/2$. Let $P$ be a $c$-balanced Bayesian network on $\{0, 1\}^d$ with known structure $G$.
Assume each parental configuration of $P$ occurs with probability at least $\alpha \ge \Omega(\eps\sqrt{\log(1/\eps)} / c)$.
Let $S'=S \cup E \setminus L$ be a set of samples such that $S$ is $\eps$-good for $P$ and $|E|+|L| \leq 2 \eps |S'|$.
There is an algorithm that, given $G$, $\eps$, and $S'$, runs in time $\wt O(d |S'|)$, and either
\begin{enumerate}
\setlength{\itemindent}{-.1in}
\item[(i)] Outputs a Bayesian network $Q$ with $\dtv(P,Q) \leq \eps\sqrt{\ln(1/\eps)}/ (c \alpha)$, or
\item[(ii)] Returns an $S''=S \cup E' \setminus L'$ such that $|S''| \le (1-\frac{\eps}{d\ln d})|S'|$ and $|E'|+|L'| < |E|+|L|$.
  \end{enumerate}
\end{proposition}
If this algorithm produces a subset $S''$, then we iterate using $S''$ in place of $S'$.
We will present the algorithm establishing Proposition \ref{prop:bal} in the following section.
We first use it to prove Theorem~\ref{thm:bal-bn}. 
\begin{proof}[{\bf Proof of Theorem~\ref{thm:bal-bn}}]
First a set $S$ of $N = \wt \Omega(m \log(1/\tau)/\eps^2)$ samples are drawn from $P$.
We assume the set $S$ is $\eps$-good for $P$\ifthenelse{\boolean{hasAppendix}}{ (which happens with probability at least $1-\tau$ by Lemma~\ref{lem:good-set})}{}.
Then an $\eps$-fraction of these samples are adversarially corrupted, giving a set $S'=S \cup E \setminus L$ with $|E|,|L| \leq \eps |S'|$.
Thus $S'$ satisfies the conditions of Proposition \ref{prop:bal}, and the algorithm outputs a smaller set $S''$ of samples that also satisfies the conditions of the proposition, or else outputs a Bayesian network $Q$ with small $\dtv(P,Q)$ that satisfies Theorem~\ref{thm:bal-bn}.
Since $|S'|$ decreases if we produce a filter, eventually we must output a Bayesian network.

Next we analyze the running time.
Observe that we can filter out at most $2\eps N$ samples, because we reject more bad samples than good ones.
By Proposition~\ref{prop:bal}, every time we produce a filter, we remove at least $\wt \Omega(d/\eps) |S'| = \wt \Omega(Nd/\eps)$ samples.
Therefore, there are at most $\wt O(d)$ iterations, and each iteration takes time $\wt O(d |S'|) = \wt O(N d)$ by Proposition~\ref{prop:bal}, so the overall running time is $\wt O(Nd^2)$.
\end{proof}

\subsection{Algorithm {\tt Filter-Known-Topology}} \label{sec:algo-description}
In this section, we present Algorithm~\ref{alg:filter} that establishes Proposition~\ref{prop:bal}.~\footnote{
We use $X \in_u S$ to denote that the point $X$ is drawn uniformly from the set of samples $S$.}

\begin{algorithm}[ht]
   \caption{\tt Filter-Known-Topology}
   \label{alg:filter}
\begin{algorithmic}[1]
   \STATE {\bfseries Input:} The dependency graph $G$ of $P$, $\eps>0$, and a (possibly corrupted) set of samples $S'$ from~$P$.    $S'$ satisfies that there exists an $\eps$-good $S$ with $S' = S \cup E \setminus L$ and $|E|+|L| \leq 2\eps|S'|$.
   \STATE {\bfseries Output:} A Bayes net $Q$ or a subset $S'' \subset S'$ that satisfies Proposition~\ref{prop:bal}.
   \STATE Compute the empirical conditional probabilities $q(i,a) = \Pr_{X \in_u S'}\left[X_i=1 \mid \Pi_{i,a}\right]$.
   \STATE Compute the empirical minimum parental configuration probability $\alpha=\min_{(i,a)} \Pr_{S'}[\Pi_{(i,a)}]$.
   \STATE Define $F(X, q)$: If $x \in \Pi_{i,a}$ then $F(x,q)_{i,a}=x_i$, otherwise $F(x,q)_{i,a}=q_{i,a}$ (Definition~\ref{def:Fxq}).
   \STATE Compute the empirical second-moment matrix of $F(X, q) - q$ and zero its diagonal, i.e., $M \in \R^{m \times m}$ with $M_{k,k}=0$, and $M_{k,\ell} = \E_{X \in_u S'}[ (F(X,q)_k - q_k) (F(X,q)_\ell - q_\ell)^T]$ for $k \ne \ell$.
   \STATE Compute the largest (in absolute value) eigenvalue $\lambda^\ast$ of $M$, and the associated eigenvector $v^{\ast}$.
   \IF{$|\lambda^\ast| \leq O(\eps \log(1/\eps)/\alpha)$}
   \STATE Return $Q = $ the Bayes net with graph $G$ and conditional probabilities $q$.
   \ELSE
   \STATE Let $\delta := 3 \sqrt{ \eps |\lambda^\ast|}/\alpha$.  Pick any $T>0$ that satisfies \label{step:findT}
     \[
      \Pr_{X \in_u S'} [|v^* \cdot (F(X,q)-q)| > T+\delta] > 7\exp(-T^2/2)+3\epsilon^2/(T^2 \ln d) \;.
     \]
     Return $S'' = $ the set of samples $x \in S'$ with $|v\cdot (F(x,q)-q) | \leq T+\delta$.
        \ENDIF
\end{algorithmic}
\end{algorithm}
At a high level, Algorithm~\ref{alg:filter} computes a matrix $M$, and shows that: either
  $\|M\|_2$ is small, and we can output the empirical conditional probabilities, or
  $\|M\|_2$ is large, and we can use the top eigenvector of $M$ to remove bad samples.

\paragraph{Setup and Structural Lemmas.} \label{ssec:setup}
In order to understand the second-moment matrix with zeros on the diagonal, $M$,
we will need to break down this matrix in terms of several related matrices, where the expectation is taken over different sets.
For a set $D=S',S,E$ or $L$, we use $w_D = |D|/|S'|$ to denote the fraction of the samples in~$D$.
Moreover, we use $M_D=\E_{X \in_u D}[((F(X,q)-q)(F(X,q)-q)^T]$ to denote the second-moment matrix of samples in $D$, and let $M_{D,0}$ 
be the matrix we get from zeroing out the diagonals of $M_D$.
Under this notation, we have $M_{S'} = w_S M_{S} + w_E M_{E} - w_L M_{L}$ and $M=M_{S',0}$.

Our first step is to analyze the spectrum of $M$, and in particular show that $M$ is close in spectral norm to $w_E M_E$.
To do this, we begin by showing that the spectral norm of $M_{S,0}$ is relatively small.
Since $S$ is good, we have bounds on the second moments $F(X,p)$.
We just need to deal with the error from replacing $p$ with $q$\ifthenelse{\boolean{hasAppendix}}{ (see Appendix~\ref{app:setup} for the proof)}{}:

\begin{lemma} \label{lem:MP-bound}
$\|M_{S,0}\|_2 \leq O(\eps + \sqrt{\sum_k \Pr_S[\Pi_k] (p_k-q_k)^2} + \sum_k \Pr_S[\Pi_k] (p_k-q_k)^2)$.
\end{lemma}

Next, we wish to bound the contribution to $M$
coming from the subtractive error. We show that this is small
due to concentration bounds on $P$ and hence on $S$. 
The idea is that for any unit vector $v$, we have tail bounds for the random variable 
$v \cdot (F(X,q)-q)$ and, since $L$ is a subset of $S$,
$L$ can at worst consist of a small fraction of the tail of this distribution. 

\begin{lemma} \label{lem:ML-bound}
$w_L \|M_L\|_2 \leq O(\eps\log(1/\eps)+\eps\|p-q\|_2^2)$.
\end{lemma}

Finally, combining the above results,
since $M_{S}$ and $M_L$ have small contribution to the spectral norm of $M$ when $\|p-q\|_2$ is small,
most of it must come from $M_E$.

\begin{lemma} \label{lem:MApprox}
$\|M - w_E M_E\|_2 \leq O\left(\eps \log(1/\eps) + \sqrt{\sum_k \Pr_{S'}[\Pi_k] (p_k-q_k)^2} + \sum_k \Pr_{S'}[\Pi_k] (p_k-q_k)^2\right)$.
\end{lemma}

Lemma~\ref{lem:MApprox} follows using the identity $|S'|M=|S| M_{S,0} + |E| M_{E,0} - |L| M_{L,0}$
and bounding the errors due to the diagonals of $M_E$ and $M_L$.

\paragraph{The Case of Small Spectral Norm.} \label{ssec:small-norm}
In this section, we will prove that if $\|M\|_2 = O(\eps \log(1/\eps)/\alpha)$, then we can output the empirical conditional means $q$.
Recall that $M_{S'} = \E_{X \in_u S'}[ (F(X,q)_i-q_i) (F(X,q)_j-q_j)^T]$ and $M = M_{S', 0}$.

We first show that the contributions that $L$ and $E$ make to $\E_{X \in_u S'}{[F(X,q)-q)}]$ can be bounded in terms of the spectral norms of $M_L$ and $M_E$.
It follows from the Cauchy-Schwarz inequality that:

\begin{lemma} \label{lem:wrong-mean-covar}
$\| \E_{X \in_u L}[F(X,q) - q] \|_2 \leq \sqrt{\|M_L\|_2}$ and $\|\E_{X \in_u E}[F(X,q) - q] \|_2 \leq \sqrt{\|M_E\|_2}$.
\end{lemma}

Combining with the results about these norms in Section~\ref{ssec:setup}, Lemma~\ref{lem:wrong-mean-covar} implies that if $\|M\|_2$ is small, then $q = \E_{X \in_u S'}[F(X,q)]$ is close to $\E_{X \in_u S}[F(X,q)]$, which is then necessarily close to $\E_{X \sim P}[F(X,p)] = p$.
The following lemma states that the mean of $(F(X,q)-q)$ under the good samples is close to $(p-q)$ scaled by the probabilities of parental configurations under $S'$:

\begin{lemma} \label{lem:mean-transformed}
Let $z \in \R^m$ be the vector with $z_k=\Pr_{S'}[\Pi_k] (p_k-q_k)$.
Then  $\| \E_{X \in_u S}[F(X,q)-q] - z\|_2 \leq O(\eps (1+\|p-q\|_2))$.
\end{lemma}

Note that $z$ is closely related to the total variation distance between $P$ and $Q$ (the Bayes net with conditional probabilities $q$).
We can write $(\E_{X \in_u S'}[F(X,q)]-q)$ in terms of this expectation under $S, E$, and $L$ whose distance from $q$ can be upper bounded using the previous lemmas.
Using Lemmas~\ref{lem:ML-bound},~\ref{lem:MApprox},~\ref{lem:wrong-mean-covar},~and~\ref{lem:mean-transformed}, we can bound $\|z\|_2$ in terms of $\|M\|_2$:

\begin{lemma} \label{lem:mean-dist-from-norm}
$\sqrt{\sum_k \Pr_{S'}[\Pi_k]^2 (p_k-q_k)^2} \leq 2\sqrt{\eps \|M\|_2} + O(\eps \sqrt{\log(1/\eps)+1/\alpha})$.
\end{lemma}

Lemma~\ref{lem:mean-dist-from-norm} implies that, if $\|M\|_2$ is small then so is $\sqrt{\sum_k \Pr_{S'}[\Pi_k]^2 (p_k-q_k)^2}$.
We can then use it to show that $\sqrt{\sum_k \Pr_P[\Pi_k] (p_k-q_k)^2}$ is small.
We can do so by losing a factor of $1/\sqrt{\alpha}$ to remove the square on $\Pr_{S'}[\Pi_k]$, and showing that $\min_k \Pr_S'[\Pi_k] = \Theta(\min_k \Pr_P[\Pi_k])$ when it is at least a large multiple of $\eps$.
Finally, if $\sqrt{\sum_k \Pr_P[\Pi_k] (p_k-q_k)^2}$ is small, Lemma~\ref{cor:bal-min} tells us that $\dtv(P,Q)$ is small.
This completes the proof of the first case of Proposition~\ref{prop:bal}.

\begin{corollary}[Part {\em (i)} of Proposition~\ref{prop:bal}] \label{cor:correct}
If $\|M\|_2 \leq O(\eps \log(1/\eps)/\alpha)$,
then 
\[\dtv(P,Q)=O(\eps\sqrt{\log(1/\eps)}/(c \min_k \Pr_P[\Pi_k])) \;.
\]
\end{corollary}

\paragraph{The Case of Large Spectral Norm.} \label{ssec:big-norm}

Now we consider the case when $\|M\|_2 \geq C\eps \ln(1/\eps)/\alpha$.
We begin by showing that $p$ and $q$ are not too far apart from each other.
The bound given by Lemma \ref{lem:mean-dist-from-norm} is 
now dominated by the $\|M\|_2$ term. 
Lower bounding the $\Pr_{S'}[\Pi_k]$ by $\alpha$ gives the following claim.

\begin{claim} \label{clm:delta-distance}
$\|p-q\|_2 \leq \delta := 3 \sqrt{\eps\|M\|_2}/\alpha$.
\end{claim}

Recall that $v^*$ is the largest eigenvector of $M$.
We project all the points $F(X,q)$ onto the direction of $v^*$.
Next we show that most of the variance of $(v^*\cdot (F(X,q) - q))$ comes from $E$.

\begin{claim} \label{clm:norm-mostly-E}
$v^{\ast T} (w_E M_E) v^{\ast} \geq \frac{1}{2} v^{\ast T} M v^{\ast}$.
\end{claim}
Claim~\ref{clm:norm-mostly-E} follows from the observation that $\|M - w_E M_E\|_2$ is much smaller than $\|M\|_2$.
This is obtained by substituting the bound on $\|p-q\|_2$ (in terms of $\|M\|_2$) from Claim \ref{clm:delta-distance} into the bound on $\|M - w_E M_E\|_2$ given by Lemma \ref{lem:MApprox}.

Claim~\ref{clm:norm-mostly-E} implies that the tails of $w_E E$ are reasonably thick.
In particular, we show that there must be a threshold $T>0$
satisfying the desired property in Step~9 of our algorithm.

\begin{lemma} \label{lem:exists-T-dist}
There exists a $T \geq 0$ such that
\[
\Pr_{X \in_u S'}[|v\cdot (F(X,q)-q)|>T+\delta] > 7\exp(-T^2/2)+3\epsilon/(T^2 \ln d) \;.
\]
\end{lemma}

If Lemma~\ref{lem:exists-T-dist} were not true, by integrating this tail bound, we can show that $v^{\ast T} M_E v^{\ast}$ would be small.
Therefore, Step~\ref{step:findT} of Algorithm~\ref{step:findT} is guaranteed to find some valid threshold $T>0$.

Finally, we show that the set of samples $S''$ we return after the filter 
is better than $S'$ in terms of $|L|+|E|$.
This completes the proof of the second case of Proposition~\ref{prop:bal}.

\begin{claim}[Part {\em(ii)} of Proposition~\ref{prop:bal}] \label{clm:calc}
If we write $S''=S \cup E' \setminus L'$, then $|E'|+|L'| < |E|+|L|$ and $|S''| \le (1-\frac{\eps}{d\ln d})|S'|$.
\end{claim}
Claim~\ref{prop:bal} follows from the fact that $S$ is $\eps$-good, so 
we only remove at most $(3 \exp(T^2/2) + \eps/T^2 \log d)|S|$ samples 
from $S$.
Since we remove more than twice as many samples from $S'$, most of the samples we throw away are from $E$.
Moreover, we remove at least $(1-\frac{\eps}{d\ln d})|S'|$ samples because we can show that the threshold $T$ is at most $\sqrt{d}$.

\label{apx:runtime}
\paragraph{Running Time of Our Algorithm~\ref{alg:filter}}
First, $q$ and $\alpha$ can be computed in time $O(N d)$ because each sample only affects $d$ entries of $q$.
We do not explicitly write down $F(X,q)$ or $M$.
Then, we use the power method to compute the largest eigenvalue $\lambda^*$ of $M$ and the associated eigenvector $v^*$.
In each iteration, we implement matrix-vector multiplication with $M$ by writing $M v$ as $\sum_i ((F(x_i,q)-q)^T v) (F(x_i,q)-q)$ for any vector $v \in \R^m$.
Because each $(F(x_i,q)-q)$ is $d$-sparse, computing $Mv$ takes time $O(dN)$.
The power method takes $(\log m/\eps')$ iterations to find a $(1-\eps')$-approximately largest eigenvalue.
We can set $\eps'$ to a small constant, because we can tolerate a small multiplicative error in estimating the spectral norm of $M$ and we only need an approximate top eigenvector (see, e.g., Corollary~\ref{cor:correct} and Lemma~\ref{clm:norm-mostly-E}).
Thus, the power method takes time $O(d N \log m)$.
Finally, computing $|v^* \cdot (F(x,q)-q)|$ takes time $O(dN)$, then we can sort the samples and find a threshold $T$ in time $O(N \log N)$, and throw out the samples in time $O(N)$.

\makeatletter{}
\section{Experiments}
\label{sec:experiment}
We test our algorithms using data generated from both synthetic and real-world networks
(e.g., the ALARM network~\cite{BeinlichSCC89}) with synthetic noise.
All experiments were run on a laptop with 2.6 GHz CPU and 8 GB of RAM.
We found that our algorithm achieves the smallest error consistently in all trials,
and that the error of our algorithm almost matches the error of the empirical conditional probabilities 
of the uncorrupted samples.
Moreover, our algorithm can easily scale to thousands of dimensions with millions of samples.~\footnote{
The bottleneck of our algorithm is fitting millions of samples of thousands dimension all in the memory.}

\subsection{Synthetic Experiments}
The results of our synthetic experiments are shown in Figure~\ref{fig:experiment-synthetic}.
In the synthetic experiment, we set $\eps = 0.1$ and first generate a Bayes net $P$ with $100 \le m \le 1000$ parameters.
We then generate $N = \frac{10m}{\eps^2}$ samples, 
where a $(1-\eps)$-fraction of the samples come from the ground truth $P$,
and the remaining $\eps$-fraction come from a noise distribution.
The goal is to output a Bayes net $Q$ that minimizes $\dtv(P, Q)$.
\ifthenelse{\boolean{hasAppendix}}{
Since there is no closed-form expression for computing the total variation distance 
between two Bayesian networks, we use sampling to estimate $\dtv(P, Q)$ 
in our experiments (see Appendix~\ref{apx:exp-dtv} for more details).}{}

\begin{figure}[ht]
\centering
\includegraphics[trim={0.6cm 0 8.0cm 1cm},clip,width=0.45\linewidth]{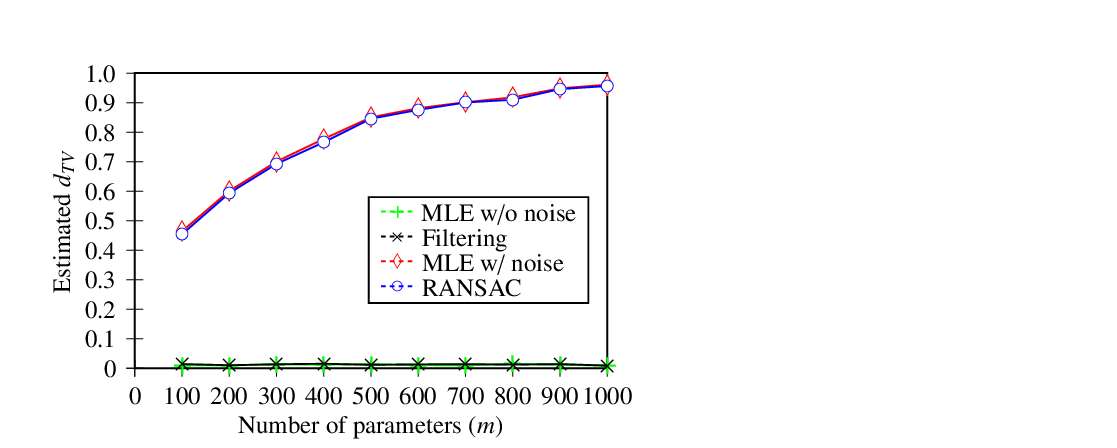} \qquad
\includegraphics[trim={0.6cm 0 8.0cm 1cm},clip,width=0.45\linewidth]{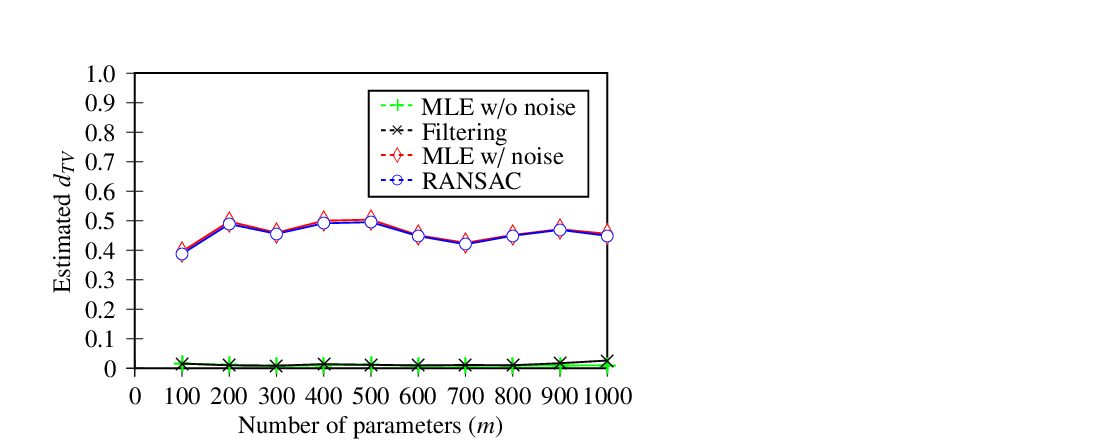}
\caption{Experiments with synthetic data: error is reported against the size of the conditional probability table (lower is better).
The error is the estimated total variation distance to the ground truth Bayes net.
We use the error of MLE without noise as our benchmark.
We plot the performance of our algorithm (\texttt{Filtering}), empirical mean with noise (\texttt{MLE}), and \texttt{RANSAC}.
We report two settings: the underlying structure of the Bayes net is a random tree (left) or a random graph (right).
}
\label{fig:experiment-synthetic}
\end{figure}

We draw the parameters of $P$ independently from $[0, 1/4] \cup [3/4, 1]$ uniformly at random, 
i.e., in a setting where the ``balancedness'' assumption does not hold.
Our experiments show that our filtering algorithm works very well in this setting,
even when the assumptions under which we can prove theoretical guarantees are not satisfied.
This complements our theoretical results and illustrates that our algorithm 
is not limited by these assumptions and can apply to more general settings in practice.

In Figure~\ref{fig:experiment-synthetic}, we compare the performance of
(1) our filtering algorithm,
(2) the empirical conditional probability table with noise, and 
(3) a RANSAC-based algorithm (see the end of Section~\ref{sec:experiment} for a detailed description).
We use the error of the empirical conditional mean without noise (i.e., MLE estimator with only good samples) as the gold standard,
since this is the best one could hope for even if all the corrupted samples are identified.
We tried various graph structures for the Bayes net $P$ and noise distributions,
and similar patterns arise for all of them.
In the top figure, the dependency graph of $P$ is a randomly generated tree,
and the noise distribution is a binary product distribution;
In the bottom figure, the dependency graph of $P$ is a random graph,
and the noise distribution is the tree Bayes net used as the ground truth 
in the first experiment.
\ifthenelse{\boolean{hasAppendix}}{The reader is referred to Appendix~\ref{apx:exp-noise} for a full description of how we generate the dependency graphs and noise distributions.}{}

\subsection{Semi-Synthetic Experiments}

In the semi-synthetic experiments, we apply our algorithm to robustly learn real-world Bayesian networks.
The ALARM network~\cite{BeinlichSCC89} is a classic Bayes net that implements a medical diagnostic system for patient monitoring.

Our experimental setup is as follows:
The underlying graph of ALARM has $37$ nodes and $509$ parameters.
Since the variables in ALARM can have up to $4$ different values,
we first transform it into an equivalent binary-valued Bayes net\ifthenelse{\boolean{hasAppendix}}{(see Appendix~\ref{apx:exp-multi-binary} for more details)}{}.
After the transformation, the network has $d = 61$ nodes and $m = 820$ parameters.
We are interested in whether our filtering algorithm can learn a Bayes net 
that is ``close'' to ALARM when samples are corrupted; 
and how many corrupted samples can our algorithm tolerate.
For $\eps = [0.05, 0.1, \ldots, 0.4]$, we draw $N = 10^6$ samples, 
where a $(1-\eps)$-fraction of the samples come from ALARM, 
and the other $\eps$-fraction comes from a noise distribution.

\begin{figure}[ht]
\begin{minipage}[c]{0.45\linewidth}
  \includegraphics[trim={0.6cm 0.2cm 8.0cm 1.1cm},clip,width=\linewidth]{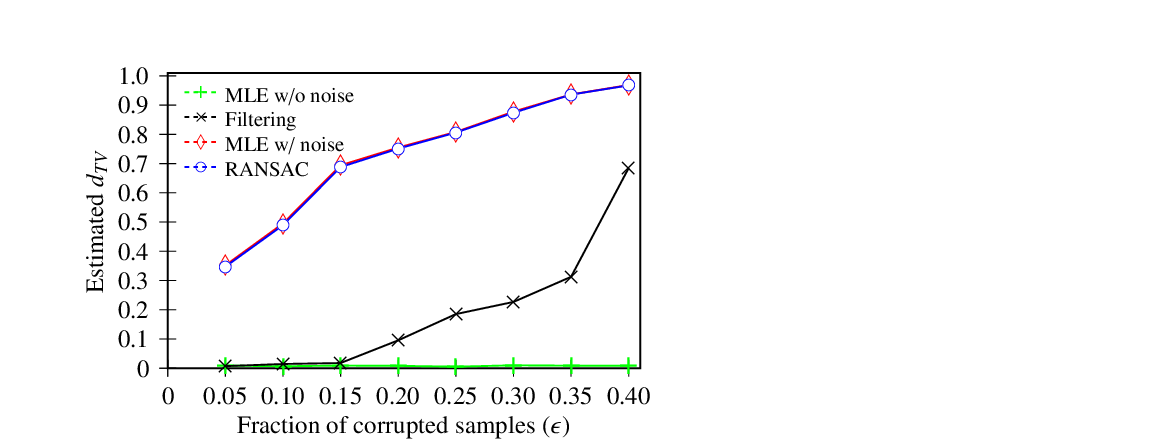}
\end{minipage}\hfill
\begin{minipage}[c]{0.52\linewidth}
  \caption{Experiments with semi-synthetic data: error is reported against the fraction of corrupted samples (lower is better).
  The error is the estimated total variation distance to the ALARM network.
  We use the sampling error without noise as a benchmark, 
  and compare the performance of our algorithm (\texttt{Filtering}), empirical mean with noise (\texttt{MLE}), and \texttt{RANSAC}.
  }
\label{fig:experiment-alarm}
\end{minipage}
\end{figure}

In Figure~\ref{fig:experiment-alarm},
we compare the performance of (1) our filtering algorithm, 
(2) the empirical conditional means with noise, 
and (3) a RANSAC-based algorithm.
We use the error of the empirical conditional means without noise as the gold standard.
We tried various noise distributions and observed similar patterns.
In Figure~\ref{fig:experiment-alarm},
the noise distribution is a Bayes net with random dependency graphs 
and conditional probabilities drawn from $[0, \frac{1}{4}] \cup [\frac{3}{4}, 1]$ 
(same as the ground-truth Bayes net in Figure~\ref{fig:experiment-synthetic}).

The experiments show that our filtering algorithm outperforms \texttt{MLE} and \texttt{RANSAC}, 
and that the error of our algorithm degrades gracefully as $\eps$ increases.
It is worth noting that even the ALARM network does not satisfy our balancedness assumption on the parameters, 
our algorithm still performs well on it and recovers the conditional probability table of ALARM in the presence of corrupted samples.

\ifthenelse{\boolean{hasAppendix}}{
\paragraph{Details of the RANSAC Algorithm.}
RANSAC uses subsampling in the hope of getting an estimator that is not affected too much by the noise.
The hope is that a small subsample might not contain any erroneous points.
The approach proceeds by computing many such estimators and appropriately selecting the best one.

In our experiments, we let RANSAC select $10\%$ of the samples uniformly at random, and repeat this process $100$ times.
After a subset of samples are selected, we compute the empirical conditional means 
and estimate the total variation distance between the corresponding Bayes net and the ground truth.
Since we know the ground truth, we can make it easier for RANSAC 
by selecting the best hypothesis that it ever produced during its execution.

The main conceptual message of our experimental evaluation of
RANSAC is that it does not perform well in high dimensions for the following reason: 
To guarantee that there are very few noisy points for such a subsample, we must take an exponential (in the dimension) number of subsets.
We are not the first to observe that RANSAC does not work for robustly learning high-dimensional distributions.
Previously, \cite{DiakonikolasKKLMS17} showed that RANSAC does not work in practice for the problem of 
robustly learning a spherical Gaussian.}{} 

\makeatletter{}\section{Conclusions and Future Directions} \label{sec:concl}
In this paper, we initiated the study of the efficient robust learning for graphical models.
We described a computationally efficient algorithm for robustly learning Bayesian networks
with a known topology, under some mild assumptions on the conditional probability table.
We evaluate our algorithm experimentally, and we view our experiments as a proof of concept demonstration that our techniques can be practical for learning fixed-structure Bayesian networks.
A challenging open problem is to generalize our results to the case when the underlying directed graph is unknown.

This work is part of a broader agenda of systematically investigating the
robust learnability of high-dimensional structured probability distributions.
There is a wealth of natural probabilistic models that merit investigation in the robust setting, 
including undirected graphical models (e.g., Ising models), and
graphical models with hidden variables (i.e., incorporating latent structure).

\paragraph{Acknowledgements.}
We are grateful to Daniel Hsu for suggesting the model of Bayes nets, and for pointing us to~\cite{Dasgupta97}.
Yu Cheng is supported in part by NSF CCF-1527084, CCF-1535972, CCF-1637397, CCF-1704656, IIS-1447554, and NSF CAREER Award CCF-1750140.
Ilias Diakonikolas is supported by NSF CAREER Award CCF-1652862 and a Sloan Research Fellowship.
Daniel Kane is supported by NSF CAREER Award CCF-1553288 and a Sloan Research Fellowship.

\newpage

\bibliographystyle{alpha}
\bibliography{allrefs}

\newcommand{\etalchar}[1]{$^{#1}$}
\begin{thebibliography}{DKK{\etalchar{+}}18b}

\bibitem[ADLS16]{AcharyaDLS16}
J.~Acharya, I.~Diakonikolas, J.~Li, and L.~Schmidt.
\newblock Fast algorithms for segmented regression.
\newblock In {\em Proceedings of the 33nd International Conference on Machine
  Learning, {ICML} 2016}, pages 2878--2886, 2016.

\bibitem[ADLS17]{ADLS17}
J.~Acharya, I.~Diakonikolas, J.~Li, and L.~Schmidt.
\newblock Sample-optimal density estimation in nearly-linear time.
\newblock In {\em Proceedings of the Twenty-Eighth Annual {ACM-SIAM} Symposium
  on Discrete Algorithms, {SODA} 2017}, pages 1278--1289, 2017.

\bibitem[AHHK12]{AnandkumarHHK12}
A.~Anandkumar, D.~J. Hsu, F.~Huang, and S.~Kakade.
\newblock Learning mixtures of tree graphical models.
\newblock In {\em Proc. 27th Annual Conference on Neural Information Processing
  Systems (NIPS)}, pages 1061--1069, 2012.

\bibitem[AKN06]{Abbeel:2006}
P.~Abbeel, D.~Koller, and A.~Y. Ng.
\newblock Learning factor graphs in polynomial time and sample complexity.
\newblock {\em J. Mach. Learn. Res.}, 7:1743--1788, 2006.

\bibitem[BDLS17]{BalakrishnanDLS17}
S.~Balakrishnan, S.~S. Du, J.~Li, and A.~Singh.
\newblock Computationally efficient robust sparse estimation in high
  dimensions.
\newblock In {\em Proc.\ 30th Annual Conference on Learning Theory (COLT)},
  pages 169--212, 2017.

\bibitem[Ber06]{Bernholt}
T.~Bernholt.
\newblock Robust estimators are hard to compute.
\newblock Technical report, University of Dortmund, Germany, 2006.

\bibitem[BGS14]{BreslerGS14a}
G.~Bresler, D.~Gamarnik, and D.~Shah.
\newblock Structure learning of antiferromagnetic {I}sing models.
\newblock In {\em NIPS}, pages 2852--2860, 2014.

\bibitem[BMS13]{BreslerMS13}
G.~Bresler, E.~Mossel, and A.~Sly.
\newblock Reconstruction of {M}arkov random fields from samples: Some
  observations and algorithms.
\newblock {\em {SIAM} J. Comput.}, 42(2):563--578, 2013.

\bibitem[Bre15]{Bresler15}
G.~Bresler.
\newblock Efficiently learning {I}sing models on arbitrary graphs.
\newblock In {\em Proc.\ 47th Annual ACM Symposium on Theory of Computing
  (STOC)}, pages 771--782, 2015.

\bibitem[BSCC89]{BeinlichSCC89}
I.~A. Beinlich, H.~J. Suermondt, R.~M. Chavez, and G.~F. Cooper.
\newblock {\em The ALARM Monitoring System: A Case Study with two Probabilistic
  Inference Techniques for Belief Networks}.
\newblock Springer, 1989.

\bibitem[CDKS17]{CanonneDKS17}
C.~L. Canonne, I.~Diakonikolas, D.~M. Kane, and A.~Stewart.
\newblock Testing {B}ayesian networks.
\newblock In {\em Proc.\ 30th Annual Conference on Learning Theory (COLT)},
  pages 370--448, 2017.

\bibitem[CDSS13]{CDSS13}
S.~Chan, I.~Diakonikolas, R.~Servedio, and X.~Sun.
\newblock Learning mixtures of structured distributions over discrete domains.
\newblock In {\em Proc.\ 24th Annual Symposium on Discrete Algorithms (SODA)},
  pages 1380--1394, 2013.

\bibitem[CDSS14a]{CDSS14}
S.~Chan, I.~Diakonikolas, R.~Servedio, and X.~Sun.
\newblock Efficient density estimation via piecewise polynomial approximation.
\newblock In {\em Proc.\ 46th Annual ACM Symposium on Theory of Computing
  (STOC)}, pages 604--613, 2014.

\bibitem[CDSS14b]{CDSS14b}
S.~Chan, I.~Diakonikolas, R.~Servedio, and X.~Sun.
\newblock Near-optimal density estimation in near-linear time using
  variable-width histograms.
\newblock In {\em Proc. 29th Annual Conference on Neural Information Processing
  Systems (NIPS)}, pages 1844--1852, 2014.

\bibitem[CGR15]{CGR15}
M.~Chen, C.~Gao, and Z.~Ren.
\newblock Robust covariance and scatter matrix estimation under {H}uber's
  contamination model.
\newblock {\em CoRR}, abs/1506.00691, 2015.

\bibitem[CGR16]{CGR15b}
M.~Chen, C.~Gao, and Z.~Ren.
\newblock A general decision theory for {H}uber's $\epsilon$-contamination
  model.
\newblock {\em Electronic Journal of Statistics}, 10(2):3752--3774, 2016.

\bibitem[CL68]{Chow68}
C.~Chow and C.~Liu.
\newblock Approximating discrete probability distributions with dependence
  trees.
\newblock {\em IEEE Trans. Inf. Theor.}, 14(3):462--467, 1968.

\bibitem[Das97]{Dasgupta97}
S.~Dasgupta.
\newblock The sample complexity of learning fixed-structure {B}ayesian
  networks.
\newblock {\em Machine Learning}, 29(2-3):165--180, 1997.

\bibitem[DDS14]{DDS14-toc}
C.~Daskalakis, I.~Diakonikolas, and R.~A. Servedio.
\newblock Learning $k$-modal distributions via testing.
\newblock {\em Theory of Computing}, 10(20):535--570, 2014.

\bibitem[DKK{\etalchar{+}}16]{DiakonikolasKKL16}
I.~Diakonikolas, G.~Kamath, D.~M. Kane, J.~Li, A.~Moitra, and A.~Stewart.
\newblock Robust estimators in high dimensions without the computational
  intractability.
\newblock In {\em Proc.\ 57th IEEE Symposium on Foundations of Computer Science
  (FOCS)}, 2016.

\bibitem[DKK{\etalchar{+}}17]{DiakonikolasKKLMS17}
I.~Diakonikolas, G.~Kamath, D.~M. Kane, J.~Li, A.~Moitra, and A.~Stewart.
\newblock Being robust (in high dimensions) can be practical.
\newblock In {\em Proc.\ 34th International Conference on Machine Learning
  (ICML)}, pages 999--1008, 2017.

\bibitem[DKK{\etalchar{+}}18a]{DiakonikolasKKLMS18}
I.~Diakonikolas, G.~Kamath, D.~M. Kane, J.~Li, A.~Moitra, and A.~Stewart.
\newblock Robustly learning a {G}aussian: Getting optimal error, efficiently.
\newblock In {\em Proc.\ 29th ACM-SIAM Symposium on Discrete Algorithms
  (SODA)}, 2018.

\bibitem[DKK{\etalchar{+}}18b]{DiakonikolasKKLSS2018sever}
I.~Diakonikolas, G.~Kamath, D.~M Kane, J.~Li, J.~Steinhardt, and A.~Stewart.
\newblock Sever: A robust meta-algorithm for stochastic optimization.
\newblock {\em arXiv preprint arXiv:1803.02815}, 2018.

\bibitem[DKS17]{DKS17-sq}
I.~Diakonikolas, D.~M. Kane, and A.~Stewart.
\newblock Statistical query lower bounds for robust estimation of
  high-dimensional {Gaussians} and {Gaussian} mixtures.
\newblock In {\em Proc.\ 58th IEEE Symposium on Foundations of Computer Science
  (FOCS)}, pages 73--84, 2017.

\bibitem[DKS18a]{DiakonikolasKS18-nasty}
I.~Diakonikolas, D.~M. Kane, and A.~Stewart.
\newblock Learning geometric concepts with nasty noise.
\newblock In {\em Proc.\ 50th Annual ACM Symposium on Theory of Computing
  (STOC)}, pages 1061--1073, 2018.

\bibitem[DKS18b]{DiakonikolasKS18-mixtures}
I.~Diakonikolas, D.~M. Kane, and A.~Stewart.
\newblock List-decodable robust mean estimation and learning mixtures of
  spherical {Gaussians}.
\newblock In {\em Proc.\ 50th Annual ACM Symposium on Theory of Computing
  (STOC)}, pages 1047--1060, 2018.

\bibitem[DKS18c]{DKS18-lr}
I.~Diakonikolas, W.~Kong, and A.~Stewart.
\newblock Efficient algorithms and lower bounds for robust linear regression.
\newblock {\em CoRR}, abs/1806.00040, 2018.

\bibitem[DLS18]{DLS18}
I.~Diakonikolas, J.~Li, and L.~Schmidt.
\newblock Fast and sample near-optimal algorithms for learning multidimensional
  histograms.
\newblock In {\em Conference On Learning Theory, {COLT} 2018}, pages 819--842,
  2018.

\bibitem[DSA11]{RQS11}
R.~Daly, Q.~Shen, and S.~Aitken.
\newblock Learning {B}ayesian networks: approaches and issues.
\newblock {\em The Knowledge Engineering Review}, 26:99--157, 2011.

\bibitem[HL18]{HopkinsL18}
S.~B. Hopkins and J.~Li.
\newblock Mixture models, robustness, and sum of squares proofs.
\newblock In {\em Proc.\ 50th Annual ACM Symposium on Theory of Computing
  (STOC)}, pages 1021--1034, 2018.

\bibitem[HM13]{HardtM13}
M.~Hardt and A.~Moitra.
\newblock Algorithms and hardness for robust subspace recovery.
\newblock In {\em Proc.\ 26th Annual Conference on Learning Theory (COLT)},
  pages 354--375, 2013.

\bibitem[HR09]{Huber09}
P.~J. Huber and E.~M. Ronchetti.
\newblock {\em Robust statistics}.
\newblock Wiley New York, 2009.

\bibitem[HRRS86]{HampelEtalBook86}
F.~R. Hampel, E.~M. Ronchetti, P.~J. Rousseeuw, and W.~A. Stahel.
\newblock {\em Robust statistics: The approach based on influence functions}.
\newblock Wiley New York, 1986.

\bibitem[Hub64]{Huber64}
P.~J. Huber.
\newblock Robust estimation of a location parameter.
\newblock {\em Ann. Math. Statist.}, 35(1):73--101, 03 1964.

\bibitem[JN07]{Jensen:2007}
F.~V. Jensen and T.~D. Nielsen.
\newblock {\em {B}ayesian Networks and Decision Graphs}.
\newblock Springer Publishing Company, Incorporated, 2nd edition, 2007.

\bibitem[JP78]{JP:78}
D.~S. Johnson and F.~P. Preparata.
\newblock The densest hemisphere problem.
\newblock {\em Theoretical Computer Science}, 6:93--107, 1978.

\bibitem[KF09]{Koller:2009}
D.~Koller and N.~Friedman.
\newblock {\em Probabilistic Graphical Models: Principles and Techniques -
  Adaptive Computation and Machine Learning}.
\newblock The MIT Press, 2009.

\bibitem[KKM18]{KlivansKM18}
A.~Klivans, P.~Kothari, and R.~Meka.
\newblock Efficient algorithms for outlier-robust regression.
\newblock In {\em Proc.\ 31st Annual Conference on Learning Theory (COLT)},
  pages 1420--1430, 2018.

\bibitem[KSS18]{KothariSS18}
P.~K. Kothari, J.~Steinhardt, and D.~Steurer.
\newblock Robust moment estimation and improved clustering via sum of squares.
\newblock In {\em Proc.\ 50th Annual ACM Symposium on Theory of Computing
  (STOC)}, pages 1035--1046, 2018.

\bibitem[LRV16]{LaiRV16}
K.~A. Lai, A.~B. Rao, and S.~Vempala.
\newblock Agnostic estimation of mean and covariance.
\newblock In {\em Proc.\ 57th IEEE Symposium on Foundations of Computer Science
  (FOCS)}, 2016.

\bibitem[LSLC18]{LSLC18-sparse}
L.~Liu, Y.~Shen, T.~Li, and C.~Caramanis.
\newblock High dimensional robust sparse regression.
\newblock {\em CoRR}, abs/1805.11643, 2018.

\bibitem[LW12]{LohW12}
P.~L. Loh and M.~J. Wainwright.
\newblock Structure estimation for discrete graphical models: Generalized
  covariance matrices and their inverses.
\newblock In {\em NIPS}, pages 2096--2104, 2012.

\bibitem[Nea03]{Neapolitan:2003}
R.~E. Neapolitan.
\newblock {\em Learning {B}ayesian Networks}.
\newblock Prentice-Hall, Inc., 2003.

\bibitem[PSBR18]{PrasadSBR2018}
A.~Prasad, A.~S. Suggala, S.~Balakrishnan, and P.~Ravikumar.
\newblock Robust estimation via robust gradient estimation.
\newblock {\em arXiv preprint arXiv:1802.06485}, 2018.

\bibitem[SW12]{SanthanamW12}
N.~P. Santhanam and M.~J. Wainwright.
\newblock Information-theoretic limits of selecting binary graphical models in
  high dimensions.
\newblock {\em {IEEE} Trans. Information Theory}, 58(7):4117--4134, 2012.

\bibitem[Tuk75]{Tukey75}
J.~W. Tukey.
\newblock Mathematics and the picturing of data.
\newblock In {\em Proceedings of the International Congress of Mathematicians},
  volume~6, pages 523--531, 1975.

\bibitem[WJ08]{Wainwright:2008}
M.~J. Wainwright and M.~I. Jordan.
\newblock Graphical models, exponential families, and variational inference.
\newblock {\em Found. Trends Mach. Learn.}, 1(1-2):1--305, 2008.

\bibitem[WRL06]{WainwrightRL06}
M.~J. Wainwright, P.~Ravikumar, and J.~D. Lafferty.
\newblock High-dimensional graphical model selection using $\ell_1$-regularized
  logistic regression.
\newblock In {\em Proc. 20th Annual Conference on Neural Information Processing
  Systems (NIPS)}, pages 1465--1472, 2006.

\end{thebibliography}

\ifthenelse{\boolean{hasAppendix}}{
\clearpage
\appendix
\makeatletter{}
\section{Omitted Proofs from Section~\ref{sec:prelims}} \label{app:prelims}

In this section, we give proofs for the technical lemmas in Section~\ref{sec:prelims}.
Lemma~\ref{cor:bal-min} bounds the total variation distance between two balanced Bayesian networks in terms of their conditional probability tables.
Lemma~\ref{cor:bal-min} is a simple corollary of Lemma~\ref{lem:hel-bn}.

\begin{lemma} \label{lem:hel-bn}
Let $P$ and $Q$ be Bayesian networks with the same dependency graph $G$.
In terms of the conditional probability tables $p$ and $q$ of $P$ and $Q$, we have:
\[
\left(\dtv(P,Q)\right)^2 \leq 2 \sum_{k} \sqrt{\Pr_P[\Pi_k] \Pr_Q[\Pi_k]} \frac{(p_k-q_k)^2}{(p_k+q_k)(2-p_k-q_k)} \;.
\]
\end{lemma}
\begin{proof}
We will in fact show this inequality for the Hellinger distance 
$\dhel(P,Q)$ and then use the standard inequality $\dtv(P,Q) \leq \sqrt{2}\dhel(P,Q)$.

Let $A$ and $B$ be two distributions on $\{0,1\}^d$. We have:
\begin{equation} \label{squared-hellinger}
1-\dhel^2(A,B) = \sum_{x \in \{0,1\}^d} \sqrt{\Pr_A[x]\Pr_B[x]} \;.
\end{equation}
Fix $i \in [d]$. The events $\Pi_{i,a}$ form a disjoint partition of $\{0,1\}^d$.
Dividing the sum above into this partition, we obtain
\begin{align}
\label{eqn:hellinger-pi-ia}
\begin{split}
1-\dhel^2(A,B) & = \sum_{a \in \{0,1\}^{|\mathrm{Parents}(i)|}} \sum_{x \in \Pi_{i,a}} \sqrt{\Pr_A[x]\Pr_B[x]} \\
& = \sum_a \sqrt{\Pr_A[\Pi_{i,a}]\Pr_B[\Pi_{i,a}]} \sum_{x \in \{0,1\}^d} \sqrt{\Pr_{A \mid \Pi_{i,a}}[x]\Pr_{B \mid \Pi_{i,a}} [x]} \; . \end{split}
\end{align}

Let $P_{\leq i}$ and $Q_{\leq i}$ be the distribution over the first $i$ coordinates of $P$
and $Q$ respectively.
Let $P_i$ and $Q_i$ be the distribution of the $i$-th coordinate of $P$ and $Q$ respectively.
 \begin{align*}
&  1-\dhel^2\left(P_{\leq i}, Q_{\leq i}\right) \\
& = \sum_a \sqrt{\Pr_{P_{\le i}}[\Pi_{i,a}] \Pr_{Q_{\le i}}[\Pi_{i,a}]} \sum_{x_{\le i}} \sqrt{\Pr_{P_{\le i} | \Pi_{i,a}}[x] \Pr_{Q_{\le i} | \Pi_{i,a}}[x]} \\
& = \sum_a \sqrt{\Pr_{P}[\Pi_{i,a}] \Pr_{Q}[\Pi_{i,a}]} \sum_{x_{\le i-1}} \sqrt{\Pr_{P_{\le i-1} | \Pi_{i,a}}[x] \Pr_{Q_{\le i-1} | \Pi_{i,a}}[x]} \sum_{x_i} \sqrt{\Pr_{P_{i} | \Pi_{i,a}}[x] \Pr_{Q_{i} | \Pi_{i,a}}[x]} \\
& = \sum_a \sqrt{\Pr_{P}[\Pi_{i,a}] \Pr_{Q}[\Pi_{i,a}]} \sum_{x_{\le i-1}} \sqrt{\Pr_{P_{\le i-1} | \Pi_{i,a}}[x] \Pr_{Q_{\le i-1} | \Pi_{i,a}}[x]} \left(1-\dhel^2\left(P_{i}\mid \Pi_{i,a}, Q_{i} \mid\Pi_{i,a}\right)\right) \\
& = \sum_a \sqrt{\Pr_{P}[\Pi_{i,a}] \Pr_{Q}[\Pi_{i,a}]} \sum_{x_{\le i-1}} \sqrt{\Pr_{P_{\le i-1} | \Pi_{i,a}}[x] \Pr_{Q_{\le i-1} | \Pi_{i,a}}[x]} \\
& \; - \sum_a \sqrt{\Pr_{P}[\Pi_{i,a}] \Pr_{Q}[\Pi_{i,a}]} \sum_{x_{\le i-1}} \left(1-\dhel^2\left(P_{\leq i-1}\mid \Pi_{i,a}, Q_{\leq i-1} \mid\Pi_{i,a}\right)\right) \dhel^2\left(P_{i}\mid \Pi_{i,a}, Q_{i} \mid\Pi_{i,a}\right) \\
& \ge 1-\dhel^2\left(P_{\leq i-1}, Q_{\leq i -1} \right) - \sum_a \sqrt{\Pr_P[\Pi_{i,a}] \Pr_Q[\Pi_{i,a}]} \; \dhel^2\left(P_{i}\mid\Pi_{i,a}, Q_{i}\mid\Pi_{i,a}\right) \;.
 \end{align*}
The first and the fifth steps use Equation~\ref{eqn:hellinger-pi-ia}, the second step uses that the $i$-th coordinate is independent of the first $(i-1)$ coordinates conditioned on $\Pi_{i,a}$, and the third and fourth steps use Equation~\ref{squared-hellinger}.

By induction on $i$, we have
\[
\dhel^2(P,Q) \leq \sum_{(i,a) \in S} \sqrt{\Pr_P[\Pi_{i,a}] \Pr_Q[\Pi_{i,a}]} \; \dhel^2\left(P_{i}\mid\Pi_{i,a}, Q_{i}\mid\Pi_{i,a}\right) \;.
\]
Now observe that the $P_{i}\mid\Pi_{i,a}$ and $Q_{i}\mid\Pi_{i,a}$ are Bernoulli distributions
with means $p_{i,a}$ and $q_{i,a}$. For $p,q \in [0,1]$, we have:
 \begin{align*}
2 \dhel^2(\textrm{Bernoulli}(p),\textrm{Bernoulli}(q)) & = (\sqrt{p}-\sqrt{q})^2 + (\sqrt{1-p}-\sqrt{1-q})^2 \\
&= (p-q)^2 \cdot \left(\frac{1}{(\sqrt{p}+\sqrt{q})^2}+\frac{1}{(\sqrt{1-p}+\sqrt{1-q})^2}\right) \\
&\le (p-q)^2 \cdot \left(\frac{1}{p+q}+\frac{1}{2-p-q}\right) \\
&= (p-q)^2 \cdot \frac{2}{(p+q)(2-p-q)} \;,
\end{align*}
and thus
\[
\dtv(P,Q)^2 \leq 2 \dhel(P,Q)^2 \leq 2 \sum_{k} \sqrt{\Pr_P[\Pi_k]\Pr_Q[\Pi_k]} \frac{(p_k-q_k)^2}{(p_k+q_k)(2-p_k-q_k)} \;. \qedhere
\]
\end{proof}

Lemma~\ref{cor:bal-min} gives a simpler expression for total variation distance between two $c$-balanced binary Bayesian networks whose minimum probability of any $\Pi_k$ is at least $\eps$.

\vspace{\topsep}
\noindent {\bf Lemma~\ref{cor:bal-min}.}
{\em Suppose that: (i) $\min_{k} \Pr_P[\Pi_k] \geq \eps$, and (ii) $P$ or $Q$ is $c$-balanced,
and \\ (iii) $\frac{3}{c}\sqrt{\sum_k \Pr_P[\Pi_k] (p_k-q_k)^2} \leq \eps$.
Then we have that $\dtv(P,Q) \leq \eps.$
}
\begin{proof}
When either $P$ or $Q$ is $c$-balanced the denominators
in Lemma \ref{lem:hel-bn} satisfies $(p_k+q_k)(2-p_k-q_k) \geq c$
and so we have $\dtv(P,Q) \leq \frac{2}{c}\sqrt{\sum_k \sqrt{\Pr_P[\Pi_k]\Pr_Q[\Pi_k]} (p_k-q_k)^2}$.
Because we assume Conditions~{\em (i)}~and~{\em (iii)}, it suffices to show that $\Pr_Q[\Pi_k] \leq \Pr_P[\Pi_k]+\eps$, which implies $\Pr_Q[\Pi_k] \leq 2\Pr_P[\Pi_k]$ and further implies $\dtv(P,Q) \le \eps$.

We can prove $\Pr_Q[\Pi_k] \leq \Pr_P[\Pi_k]+\eps$ by induction on $i$. Suppose that for all $1 \leq j < i$ and all $a' \in \{0, 1\}^{|\text{Parents}(j)|}$,
$\Pr_Q[\Pi_{j,a'}] \leq \Pr_P[\Pi_{j,a'}]+\eps$.
Then we have $\dtv(P_{\leq i-1},Q_{\leq i-1}) \leq \eps$.
Because the event $\Pi_{i,a}$ depends only on $j < i$, we have
  $|\Pr_P[\Pi_{i,a}] - \Pr_Q[\Pi_{i,a}]| \leq \dtv(P_{\leq i-1},Q_{\leq i-1}) \leq \eps$,
  and therefore $\Pr_Q[\Pi_{i,a}] \leq \Pr_P[\Pi_{i,a}]+\eps$ for all $a$.
\end{proof}

We associate a vector $F(X)$ to each sample $X$, so that $F(X)$ contains information about the conditional means, and learning the mean of $F(X)$ to good accuracy is sufficient to recover the distribution.
Recall that $q$ is the vector of empirical conditional means, and we define $F(x,q): \{0, 1\}^d \to [0, 1]^m$ as follows (Definition~\ref{def:Fxq}):
  If $x \in \Pi_{i,a}$, then $F(x,q)_{i,a}=x_i$, otherwise $F(x,q)_{i,a}=q_{i,a}$.

We will prove some properties of $F$.
First, we note that $F$ is invertible in the following sense.

\begin{claim} \label{lem:inverse}
Fix $q \in [0,1]^m$ and $j \in [d]$.
Given $(x_1,\ldots,x_j)$, we can compute $F(x,q)_{i,a}$ for all $(i,a)$ with $i \leq j$.
We can recover $(x_1, \dots, x_j)$ from these $F(x,q)_{i,a}$ as well.
\end{claim}
\begin{proof}
By the definition of $F(x,q)$, to compute $F(x,q)_{i,a}$ we need to know $x_i$ and whether $x \in \Pi_{i,a}$.
Note that whether or not $x \in \Pi_{i,a}$ depends only on $(x_1,\dots, x_{i-1})$, so $F(x,p)_{i,a}$ is a function of $(x_1,\dots, x_{i-1}, x_i)$.

We will show by induction that $(x_1, \dots, x_j)$ can be recovered from all $F(x,p)_{i,a}$ with $i \leq j$.
Since $x_1$ has no parents, we have $x_1 = F(x,p)_{1,a'}$ for the empty bitstring $a'$.
For $i > 1$, we have that $x_i = F(x,p)_{i,a}$ for the unique $a$ with $x \in \Pi_{i,a}$, and we can decide which $a$ based on $(x_1,\dots, x_{i-1})$.
\end{proof}

Next, we show that when $q = p$ (the true conditional probabilities), 
  although the coordinates of $F(X,p)$ are not independent,
  the mean of a coordinate of $F(X,p)$ remains unchanged even if we condition on the values of previous coordinates.
\begin{claim} \label{clm:cond-mean}
We have that $\E_{X \sim P}[F(X,p)_k \mid F(X,p)_1, \dots, F(X,p)_{k-1}]=p_k$ for all $k \in [m]$.
\end{claim}
\begin{proof}
Let $k = (i,a)$. Since we order the $(i,a)$'s lexicographically, $(F(X,p)_1, \dots, F(X,p)_{k-1})$
includes $F(X)_{j,a'}$ for all $(j,a')$ with $j < i$.
By Claim~\ref{lem:inverse}, these determine the value of the parents of $X_i$,
i.e., whether or not $\Pi_{i,a}$ occurs.

If $\Pi_{i,a}$ occurs, $F(X,p)_{i,a} = X_i$ is a Bernoulli with mean $p_{i,a}$.
If $\Pi_{i,a}$ does not occur, $F(X,p)_{i,a}$ is deterministically $p_{i,a}$.
Either way, we have $\E[F(X,p)_k \mid F(X,p)_1, \ldots, F(X,p)_{k-1}]=p_k$
for all combinations of $F(X,p)_1, \ldots, F(X,p)_{k-1}$.
\end{proof}

We build on Claim~\ref{clm:cond-mean} to show that, although the coordinates of $F(X,p)$ are not independent,
the first and second moments are the same as that of a product distribution of the marginal of each coordinate.

\vspace{\topsep}
\noindent {\bf Lemma \ref{lem:F-moments}.}
{\em For $X \sim P$, we have $\E\left[ F(X,p) \right] = p$.
The covariance matrix of $F(X,p)$ satisfies $\Cov[F(X,p)]= \diag(\Pr_P[\Pi_k]p_k(1-p_k))$.
}

\begin{proof}
Note that $\E[F(X,p)]_{k} = \Pr_P[\Pi_k]p_k+(1-\Pr_P[\Pi_k])p_k=p_k$ for all $k \in [m]$.

We first show that any off-diagonal entry of the covariance matrix is $0$.
That is, for any $(i,a) \neq (j,a')$, $F(X,p)_{i,a}$ and $F(X,p)_{j,a'}$ are uncorrelated:
$\E[(F(X,p)_{i,a}-p_{i,a})(F(X,p)_{j,a'}-p_{j,a'})]=0$.
If $i=j$ then $\Pi_{i,a}$ and $\Pi_{j,a'}$ cannot simultaneously hold.
Therefore, at least one of $F(X,p)_{i,a}$ and $F(X,p)_{j,a'}$ is deterministic, so they are uncorrelated.
When $i \neq j$, we assume without loss of generality that $i > j$.
When $i > j$, we claim that conditioned on the value of $F(X,p)_{j,a'}$,
  the expected value of $F(X,p)_{i,a}$ remains the same.
In fact, Claim~\ref{clm:cond-mean} states that even after conditioning on all of $F(X,p)_{j,a'}$ with $j < i$, the expectation of $F(X,p)_{i,a}$ is $p_{i,a}$.

Finally, for any $(i,a) \in S$,
\[
\E\left[(F(X,p)_{i,a}-p_{i,a})^2\right]=\Pr_P[\Pi_{i,a}] \E\left[(X_i-p_{i,a})^2 \mid \Pi_{i,a}\right]
=\Pr_P[\Pi_{i,a}] p_{i,a}(1-p_{i,a}) \; . \qedhere
\]
\end{proof}

Finally, we will need a suitable concentration inequality that works under conditional independence.
Lemma~\ref{lem:Azuma} shows that the projections of $F(X,q)$ on any direction $v$ is concentrated around its mean.

\vspace{\topsep}
\noindent {\bf Lemma~\ref{lem:Azuma}.}
{\em
For $X \sim P$ and any unit vector $v \in \R^d$ and any $q \in [0, 1]^m$, we have}
\[
\Pr[|v \cdot (F(X, q) - q)| \geq T + \|p-q\|_2] \leq 2\exp(-T^2/2) \;.
\]
\begin{proof}
By Claim \ref{clm:cond-mean}, $\E_{X \sim P}[F(X,p)_k \mid F(X,p)_1, \dots, F(X,p)_{k-1}]=p_k$ 
for all $1 \leq k \leq m$.
Thus, the sequence
$\sum_{k=1}^\ell v_k (F(X,p)_k - p_k)$ for $1 \leq \ell \leq m$ is a martingale,
and we can apply Azuma's inequality.
Note that the maximum valueof $v_k (F(X,p)_k - p_k)$ is $v_k$, and thus
\[
\Pr[|v \cdot (F(X,p)-p)| \geq T] \leq 2\exp(-T^2/2\|v\|_2^2) =  2\exp(-T^2/2) \;.
\]
Consider an $x \in \{0,1\}^d$.
If $x \in \Pi_{i,a}$, then we have
$F(x,p)_{i,a}=F(x,q)_{i,a}=x_i$ and so
\[
|(F(x,p)_{i,a} - p_{i,a}) - (F(x,q)_{i,a}-q_{i,a})| = |p_{i,a}-q_{i,a}| \;.
\]
If $x \notin  \Pi_{i,a}$, then $F(x,p)_{i,a} = p_{i,a}$ and $F(x,q)_{i,a} = q_{i,a}$,
hence
\[
|(F(x,p)_{i,a} - p_{i,a}) - (F(x,q)_{i,a}-q_{i,a})| = 0 \;.
\]
Thus, we have
\[
\|(F(x,p)-p) - (F(x,q)-q)\|_2 \leq \|p-q\|_2 \;.
\]
An application of the Cauchy-Schwarz inequality gives that,
  if $|v \cdot (F(X,q)-q)| \geq T + \|p-q\|_2$ then $|v \cdot (F(x,p)-p)| \geq T$.
Therefore, the probability of the former holding for $X$
must be at most the probability that the latter holds for $X$.
\end{proof}

\section{Omitted Proofs from Sections~\ref{sec:alg}} \label{app:alg}

This section analyzes Algorithm~\ref{alg:filter} and gives the proof of Proposition \ref{prop:bal}.

Recall that our algorithm takes as input an $\eps$-corrupted multiset $S'$ of $N = \wt \Omega(m \log(1/\tau) /\eps^2)$ samples from a $d$-dimensional ground-truth Bayesian network $P$.
We write $S'=(S \setminus L) \cup E$, where $S$ is the set of samples before corruption,
  $L$ contains the good samples that have been removed or (in later iterations) incorrectly rejected by filters,
  and $E$ represents the remaining corrupted samples.
We mapped each sample $X$ to a vector   $F(X,q)$ which contains information about the empirical conditional means $q$.
We are interested in the spectral norm of $M$, the empirical second-moment matrix of $F(X, q) - q$ over the set $S'$ with zeros on the diagonal: $M_{k,\ell}=0$, and $M_{k,\ell} = \E_{X \in_u S'}[ (F(X,q)_k - q_\ell) (F(X,q)_k - q_\ell)^T]$ for $k \ne \ell$.

The basic idea of the analysis is as follows:
If the empirical conditional probability table $q$ is close to 
the true conditional probability table $p$ of $P$, then outputting $q$ is correct.
We know that we have enough samples that the empirical conditional probability table with no noise
$\wt p$ is a good approximation to $p$.  
Therefore, we will be in good shape
so long as the corruption of our samples
does not introduce a large error in the conditional probability table.

Thinking more concretely about this error, we may split it into two parts:
$L$, the subtractive error, and $E$ the additive error.
Using concentration results for $P$, it can be shown that
the subtractive errors cannot cause significant problems
for the conditional probability table.
It remains to consider additive errors.
The bad samples in $E$ can introduce notable errors in the conditional
probability table, since any given sample can be $\sqrt{d}$ far
from the mean.
If many of the corrupted samples line up in the same direction,
this can lead to a notable discrepancy.

However, if many of these errors line up in some direction
(which is necessary in order to have a large impact on the mean),
the effects will be reflected in the first two moments.
More concretely, if for some unit vector $v$,
the expectation of $v\cdot F(E,q)$ is very far from
the expectation of $v\cdot F(P,q)$, this will force
the variance of $v\cdot F(S',q)$ to be large.
This implies two things: First, it tells us that if $v\cdot F(S',q)$
is small for all $v$ (a condition equivalent to $\|M\|_2$ being small),
we know that $q$ is a good approximation to the true
conditional probability table. Second, if $\|M\|_2$ is large, we can find a unit vector $v$ where $v\cdot F(S', q)$ has large variance.
A reasonable fraction of this variance must be coming from samples in $E$ that have $v\cdot F(X,q)$ very far from the mean.
On the other hand, using concentration bounds
for $v\cdot F(P,q)$, we know that very few
valid samples are this far from the mean.
This discrepancy will allow us to create a filter
  which rejects more samples from $E$ than from $S$.

In Section~\ref{app:good}, we will provide a set of deterministic conditions that we expect from the good samples and show that they happen with high probability.
In Section~\ref{app:setup}, we will prove some structural lemmas about the spectrum of $M$.
In Section~\ref{app:small-norm}, we will show that if $\|M\|_2$ is small, then we can output the empirical conditional probabilities.
In Section~\ref{app:big-norm}, we will show that if $\|M\|_2$ is large, then we can use the top eigenvector of $M$ to remove bad samples.

\subsection{Deterministic Conditions that We Require on the Good Samples} \label{app:good}

Given a large enough set $S$ of good samples drawn from the ground-truth Bayesian network $P$,
Lemma~\ref{lem:good-set} states that, for $X \in_u S$, the mean, covariance, and tail bounds of $F(X,p)$ behave like we would expect them to.
We call a set of samples that satisfies these properties $\eps$-good for $P$.

\begin{lemma} \label{lem:good-set}
Let $S$ be a set of $\Omega((m \log(m/\eps) + \log(1/\tau)) \cdot \log^2 d \cdot \eps^{-2})$ samples from $P$. 
Let $p$ and $\wt p$ denote the conditional probability tables of $P$ and of the empirical distribution given by $S$ respectively. 
Then, with probability at least $1-\tau$, we have the following:
\begin{itemize}
\item[(i)] $|\Pr_{S}[\Pi_k]-\Pr_{P}[\Pi_k]| \leq \eps$,
\item[(ii)] $\sum_{k} \Pr_{S}[\Pi_{k}] (\wt p_{k}- p_{k})^2 \leq \eps^2$,
\item[(iii)] For all unit vectors $v$ and $T > 0$, we have
\[
\Pr_{X \in_u S}[|v \cdot (F(X, p) - p)| \geq T] \leq 3 \exp(-T^2/2) + \eps/(T^2 \ln d) \;.
\]
\item[(iv)] $\|\E_{X \in_u S}[(F(X, p)-p)(F(X, p)-p)^T] - \Cov_{Y \sim P}[F(Y,p)-p]\|_2 \leq O(\eps)$. 
\item[(v)] Let $A = \E_{X \in_u S}[(F(X, p)-p)(F(X, p)-p)^T]$. If $A_0$ is the matrix obtained by zeroing the diagonal of $A$, then $\|A_0\|_2 = O(\eps)$.
\end{itemize}
\end{lemma}

\begin{proof}
For (i), by the Chernoff and union bounds, with probability at least $1-\tau/10$,
we have that our empirical estimates for $\Pr_{P}[\Pi_k]$ are correct to within $\eps$
as long as we have at least $O(\log(m/\tau)/\eps^2)$ samples.

Note that for a fixed $k=(i,a)$, we have $\wt p_k$ is the empirical expectation 
of  $\Pr_{S}[\Pi_{k}] N$ independent samples from a Bernoulli with probability $p_k$. 
By Chernoff bounds, when $N \geq \Omega(m \log(m/\tau)/\eps^2)$, we have $|\wt p_{k} - p_{k}| \leq \eps/\sqrt{m \Pr_{S}[\Pi_{k}]}$ 
with probability at least $1-\tau/10m$, . By a union bound, 
this holds for all $k$ except with probability at most $1/10\tau$. 
Then we have $\sum_{k} \Pr_{S}[\Pi_{k}] (\wt p_{k}- p_{k})^2 \leq \eps^2$.

For (iii) and (iv), we first prove this happens for a fixed $v$ and $T$ 
with sufficiently high probability and then take a union bound over a cover of $v$ and $T$.
 
\begin{claim} \label{clm:1Diii}
Let $S$ be a set of $N$ samples from $P$.
Let $X \in_u S$ and $Y \sim P$.
For any unit vector $v$ and $T \geq 0$, we have that
\begin{enumerate}
\item[(i)] $|\E[(v \cdot(F(X,p)-p))^2]-\E[(v \cdot (F(Y,p)-p))^2]| \leq O(\eps)$, and
\item[(ii)] $\Pr[|v \cdot(F(X,p) - p)| \geq T] \leq 5 \exp(-T^2/2)/2 + \eps/(2T^2)$,
\end{enumerate}
with probability at least $1-\exp(-\Omega(N\eps^2))$.
\end{claim}
\begin{proof}
For the variance, note that for $Y \sim P$, $v \cdot (F(Y,p)-p)$ 
is sub-Gaussian by Lemma \ref{lem:Azuma}. 
Thus, $N\E[(v \cdot (F(X,p)-p))^2]$ is the sum of $N$ i.i.d. squares 
of sub-Gaussian random variables. By the Hanson-Wright inequality, for any $t > 0$,
\[
\Pr[|N\E[(v \cdot(F(X,p)-p))^2] - N\E[(v \cdot (F(Y,p)-p))^2]| \geq t] \leq 2 \exp(\Omega( \min\{t^2/N,t\})) \;.
\]
Applying this with $t=N\eps$, we get that $|\E[(v \cdot(F(X,p)-p))^2]-\E[(v \cdot (F(Y,p)-p))^2]| \leq \eps$ 
except with probability at most $\exp(-\Omega( \min\{N\eps^2,N\eps\}))=\exp(-\Omega(N\eps^2))$.

By Lemma \ref{lem:Azuma}, we have $\Pr[|v \cdot (F(Y,p) - p)| \geq T] \leq 2\exp(-T^2/2)$.
Hence,
\[N\Pr_{X \in_u S}[|v \cdot (F(X,p) - p)| \geq T] \]
  is the sum of $N$ i.i.d. Bernoulli random variables,
  each with mean at most $2\exp(-T^2/2)$. We use the following two versions of the Chernoff bound:
\begin{fact}
Let $Z_1,\ldots,Z_N$ be i.i.d. Bernoullis with mean $\mu$. Then
\begin{itemize}
\item[(i)] $\Pr[\sum_i Z_i/N \geq (1+\delta)\mu] \leq \exp(-\delta\ln(1+\delta) N \mu/2)$ for $\delta > 0$.
\item[(ii)] $\Pr[\sum_i Z_i/N \geq \nu] \leq \exp(-D(\nu||\mu) N)$ 
for $\nu \geq \mu$, where $D(\nu||\mu)=\nu\ln(\nu/\mu) + (1-\nu) \ln((1-\nu)/(1-\mu))$ 
is the KL-divergence between Bernoullis with probabilities $\nu$ and $\mu$ .
\end{itemize}
\end{fact}
Here $Z_i$ is a Bernoulli random variable where $Z_i = 1$ if and only if for the $i$-th sample in $S$ we have $(v \cdot (f(X_i,p) - p) \ge T)$.
Let $\nu_1 = \nu_1(T)=5 \exp(-T^2/2)/2$, $\nu_2 = \nu_2(T)=\eps/(2T^2)$, and $\nu = \nu(T) = \nu_1(T) + \nu_2(T)$.
We want to prove that
\[ \Pr_{X \in_u S}[|v \cdot(F(X,p) - p)| \geq T] = \sum_{i=1}^N Z_i / N \ge \nu(T)\]
  happens with probability at most $\exp(-\Omega(N\eps^2))$ for any $T > 0$.

We have $\mu=\Pr_{Y \sim P}[|v \cdot (f(Y,p) - p)| \geq T] \leq2\exp(-T^2/2)$.
Let $T' = \Theta(\sqrt{\log(1/\eps)})$ be such that $2\mu(T') = 4\exp(-T'^2/2) = \eps^2/(4T'^4) = \nu_2(T')^2$.
For $T \leq T'$, we use bound (i), and for $T \geq T'$, we use bound (ii).

Since $\nu \ge \nu_1 \geq 5\mu/4$, by (i),
$\sum_i Z_i/N \geq \nu$ with probability at most 
$\exp(-N \ln(5/4)(\nu-\mu)/8) \leq \exp(-N\nu/180)$.
When $T \leq T'$, $\nu \geq \nu_2 = \eps/(2T^2) = \Omega(\eps/\log(1/\eps))$, 
and so $\sum_i Z_i/N \geq \nu_1(T)$ with probability at most $\exp(-\Omega(N\eps/\log(1/\eps)))$. 

When $T \geq T'$, we have $2 \mu(T) \leq \nu_2(T)^2$ and so $\ln(\nu_2/\mu) \geq \ln(2/\mu)/2 \geq T^2/4$. 
Thus, we have
\begin{align*}
D(\nu_2||\mu)
& = \nu_2\ln(\nu_2/\mu) + (1-\nu_2) \ln((1-\nu_2)/(1-\mu)) \\
& \geq \nu_2\ln(\nu_2/\mu) + (1-\nu_2)\ln(1-\nu_2) \\
& \geq \nu_2 \ln(\nu_2/\mu) + (1-\nu_2)(-\nu_2+O(\nu_2^2)) \\
& \geq \nu_2 \cdot (\ln(\nu_2/\mu) - 1 - O(\nu_2)) \\
& \geq \nu_2 \cdot (T^2/4 -1 - O(\eps)) \\
& \geq \nu_2 \cdot (T^2/5) = \eps/10 \;.
\end{align*}
Using bound (ii), we get that  $\Pr[|v \cdot (f(X,p) - p)| \geq T] \geq \nu(T)$ 
with probability at most $\exp(-\Omega(N\eps))$. 

In either case, we can take a union bound with the probability 
that the variance was far above and get that both requirements 
hold with probability at least $1 - \exp(-\Omega(N\eps^2))$.
\end{proof}

Now we continue to prove Conditions (iii) and (iv) of Lemma~\ref{lem:good-set}.
Let $\mathcal{C}$ be an $(\eps/d)$-cover of the unit sphere in $\R^m$ 
in Euclidean distance of size $O(d/\eps)^m$. Let $\mathcal{T}$ be all multiples of 
$\sqrt{\eps}$ that are in the interval $[0, \sqrt{d}]$. 
Thus, the number of combinations of $v$ and $T$ from both covers 
is at most $|\mathcal{C}||\mathcal{T}| \leq O(d/\eps)^{m+1}$.
When $N \geq \Omega((m \log(d/\eps) + \log(1/\tau))/\eps^2)$, by a union bound, Claim \ref{clm:1Diii}~(i) holds for all $v' \in \mathcal{C}$, $T' \in \mathcal{T}$ except with probability $\exp(O((m+1)\log(d/\eps))-\Omega(N\eps^2)) \le \tau/10$.
We assume that this happens and continue to prove (iv).

Note that for every unit vector $v \in R^m$, there exists a unit vector $v' \in \mathcal{C}$  with $\|v-v'\|_2 \leq \eps/d$.
Since for all $x \in \{0,1\}^m$, $v, v' \in \R^m$, $\|(F(x,p) - p)\|_2 \leq \sqrt{d}$, 
we have $|(v \cdot (F(x,p) - p))^2-(v' \cdot (F(x,p) - p))^2| = |(v+v') \cdot (F(x,p) - p)| |(v-v') \cdot (F(x,p) - p)| \leq 2d \|v-v'\|_2$. 
Thus, 
\begin{align*}
& |\E[(v \cdot (F(X,p)-p)^2]-\E[(v \cdot (F(Y,p)-p)^2]| \\
& \leq |\E[(v' \cdot (F(X,p)-p)^2]-\E[(v' \cdot (F(Y,p)-p)^2]| + 4 \eps = O(\eps) \;.
\end{align*}
Since this holds for every unit vector $v$, we have that $\|\E_{X \in_u S}[(F(X, p)-p)(F(X, p)-p)^T] - \Cov_{Y \sim P}[F(Y,p)-p]\|_2 \leq O(\eps)$. This is (iv).

For (iii), we will use Claim~\ref{clm:1Diii}~(ii) with $\eps' = \eps/\ln(d)$.
That is, when $N \ge \Omega((m \log(d/\eps) + \log(1/\tau))/\eps'^2) = \Omega((m \log(d/\eps) + \log(1/\tau)) \cdot \log^2 d \cdot \eps^{-2})$, for every $v' \in \mathcal{C}$, $T' \in \mathcal{T}$, we have
\[
\Pr[|v \cdot(F(X,p) - p)| \geq T] \leq 5 \exp(-T^2/2)/2 + \eps'/(2T^2) = 5 \exp(-T^2/2)/2 + \eps/(2T^2 \ln d) \; .
\]

Note that $\Pr[|v \cdot (F(X,p) - p)| \geq T]=0$ for $T > \sqrt{d}$ since $\|(F(x,p) - p)\|_2 \leq \sqrt{d}$. For $T \leq 1$, $3\exp(-T^2/2) \geq 1$ and (iii) is trivial. Given a unit vector $v \in \R^m$ and $T$ with $1 \leq T \leq \sqrt{d}$, there exists a $v' \in \mathcal{C}$ and $T' \in \mathcal{C'}$ with $\|v-v'\|_2 \leq \eps/d$ and $T^2-2\eps \leq T'^2 \leq T^2-\eps$. Note that $(T-\eps/\sqrt{d})^2 \geq T^2 - \eps \geq T'^2$. Then if $|v \cdot (F(X,p) - p)| \geq T$, then $|v' \cdot (f(X,p) - p)| \geq |v \cdot (F(X,p) - p)| -\eps/\sqrt{d} \geq T'$. Now we have
\begin{align*}
\Pr[|v \cdot (F(X,p) - p)| \geq T] & \leq \Pr[|v' \cdot (F(X,p) - p)| \geq T'] \\
& \leq 5 \exp(-T'^2/2)/2 + \eps/(2T'^2 \ln d) \\
& \leq 5 \exp(\eps-T^2/2)/2 + \eps/((2T^2 - 4\eps)\ln d)\\
& \leq 3 \exp(-T^2/2) + \eps/(T^2 \ln d) \;.
\end{align*}
This completes the proof of (iii).

Finally we prove (v). 
We claim that this follows from (i), (ii) and (iv). From (iv), we have that $\|A - \Cov_{Y \sim P}[F(Y,p)-p]\|_2 \leq O(\eps)$.
By Lemma~\ref{lem:F-moments}, $\Cov_{Y\sim P}[F(Y,p)]$ is a diagonal matrix $\diag(\Pr_P[\Pi_k]p_k(1-p_k))$.
Therefore, we need to show that the diagonal elements of $A=\E_{X \in_u S}[(F(X, p)-p)(F(X, p)-p)^T]$ and $\Cov_{Y \sim P}[F(Y,p)-p]$ are close.
Let $A_\diag$ be the diagonal matrix with the diagonal entries of $A$.
\begin{align*}
\|A_0\|_2 = \|A - A_\diag\|_2 
& \leq \|A - \Cov_{Y \sim P}[F(Y,p)-p] \|_2 + \|A_\diag - \Cov_{Y \sim P}[F(Y,p)-p]\|_2 \\
& \leq O(\eps) + \max_k |A_{k,k} - \Pr_P[\Pi_k]p_k(1-p_k)| \;.
\end{align*} 

Let $\wt p_k = \Pr_{X \in_u S}[X_i = 1 | \Pi_k]$ denote the empirical conditional means.
Consider a diagonal entry of $A$.
For $k = (i,a)$, 
$A_{k,k}=\E[(F(X,p)-p)_k^2] = \Pr_S(\Pi_k)\E[(X_i-p_k)^2|\Pi_k]= \Pr_S(\Pi_k) (\wt p_k (1-p_k)^2 + (1-\wt p_k) p_k^2)=\Pr_S(\Pi_k) ( p_k^2 + \wt p_k (1-2p_k))$. Then we have 
\begin{align*}
|A_{k,k} - \Pr_P[\Pi_k]p_k(1-p_k)|
& = |\Pr_S(\Pi_k) ( p_k^2 + \wt p_k (1-2p)) - \Pr_P[\Pi_k]p_k(1-p_k)| \\
& \leq |\Pr_S(\Pi_k) - \Pr_P[\Pi_k]|p_k(1-p_k) + \Pr_S(\Pi_k)|(\wt p_k-p_k)(1-2p_k)| \\
& \leq |\Pr_S(\Pi_k) - \Pr_P[\Pi_k]| +  \Pr_S(\Pi_k)|(\wt p_k-p_k)| \leq O(\eps) \;,
\end{align*}
assuming (i) and (ii). 
This completes the proof of (v).

By a union bound, (i)-(v) all hold simultaneously with probability at least $1-\tau$.
\end{proof}

\subsection{Omitted Proofs from Section~\ref{ssec:setup}: Setup and Structural Lemmas} \label{app:setup}

In this section, we prove some structural lemmas that we will need to prove Proposition~\ref{prop:bal}.

In order to understand the second-moment matrix with zeros on the diagonal, $M$,
we will need to break down this matrix in terms of several related matrices, where the expectation is taken over different sets.
For a set $D=S',S,E$ or $L$, we use $w_D = |D|/|S'|$ to denote the fraction of the samples in~$D$.
Moreover, we use $M_D=\E_{X \in_u D}[((F(X,q)-q)(F(X,q)-q)^T]$ to denote the second-moment matrix of samples in $D$, and let $M_{D,0}$ 
be the matrix we get from zeroing out the diagonals of $M_D$.
Under this notation, we have $M_{S'} = w_S M_{S} + w_E M_{E} - w_L M_{L}$ and $M=M_{S',0}$.

First, we note that since the probabilities of the parental configurations are probabilities, the noise will not move them much.
Abusing notation, we use $\alpha$ for the {\em empirical} minimum parental configuration.
\begin{lemma} \label{lem:alpha-good}
For all $k$,
$|\Pr_{S'}[\Pi_k]-\Pr_{S}[\Pi_k]| \leq 2\eps$
and $\alpha \geq (C'-3)\eps \geq \eps$.
\end{lemma}
\begin{proof}
Proposition~\ref{prop:bal} requires that $|E|+|L| \le 2\eps|S'|$. We have
\begin{align*}
|\Pr_{S'}[\Pi_k] - \Pr_S[\Pi_k]| & = |(w_L - w_E)\Pr_S[\Pi_k] - w_L \Pr_L[\Pi_k] + w_E\Pr_E[\Pi_k]| \leq w_L + w_E \leq 2\eps \; .
\end{align*}
Since $S$ is $\eps$-good, by Lemma~\ref{lem:good-set} (i), $|\Pr_{P}[\Pi_k] - \Pr_{S}[\Pi_k]| \leq \eps$.
Since we assume that $\min_k \Pr_P [\Pi_k] \geq 4\eps$, $\alpha = \min_k \Pr_{S'} [\Pi_k] \geq \eps$.
\end{proof}

Our next step is to analyze the spectrum of $M$, and in particular show that $M$ is close in spectral norm to $w_E M_E$.
To do this, we begin by showing that the spectral norm of $M_{S,0}$ is relatively small.
Since $S$ is good, we have bounds on the second moments $F(X,p)$.
We just need to deal with the error from replacing $p$ with $q$.

\vspace{\topsep}
\noindent {\bf Lemma \ref{lem:MP-bound}.}
{\em 
$\|M_{S,0}\|_2 \leq O(\eps + \sqrt{\sum_k \Pr_S[\Pi_k] (p_k-q_k)^2} + \sum_k \Pr_S[\Pi_k] (p_k-q_k)^2)$.
}

\begin{proof}
Let $A_S$ denote the second-moment matrix of $(F(X,p) - p)$ under $S$.
\[A_S=\E_{X \in_u S}[(F(X, p)-p)(F(X, p)-p)^T] \; .\]
Let $M_{S,\diag}$,$A_{S,\diag}$ be the matrices obtained by zeroing all the non-diagonal entries of $M_S$ and $A_S$ respectively.
We will use the triangle inequality:
\[
\| M_{S,0} \|_2 \le \| M_S - A_S \|_2 + \| M_{S,\diag} - A_{S,\diag} \|_2 + \| A_{S,0} \|_2 \; .
\]

Since $S$ is $\eps$-good, we have by Lemma~\ref{lem:good-set} (iv) that the matrix $A_S$ is $O(\eps)$ close to the diagonal matrix $ \Cov_{Y \sim P}[F(Y,p)-p]$ and by Lemma~\ref{lem:good-set} (v) that the matrix formed by zeroing the diagonal of $A_S$, $A_{S,0}$, has $\|A_{S,0}\|_2 \leq O(\eps)$.

First we will show that $M_S$ is close to $A_S$, and then we will show that their diagonals are close which implies that $M_{S,0}$ is close to $A_{S,0}$.

For notational convenience, let us define $f(x,r)=F(x,r)-r$ for all $x, r \in \R^m$.
Note that $\|A_S\|_2 \leq \|\Cov_{Y \sim P}[f(Y,p)]\|_2 + O(\eps) = \|\diag(\Pr_P[\Pi_k]p_k(1-p_k))\|_2+O(\eps) \leq 1+O(\eps)$ by Lemma~\ref{lem:F-moments}.
Let $B$ be the matrix $\E_{X \in_u S}[ (f(X,q)-f(X,p)) (f(X,q)-f(X,p))^T]$.
For any unit vector $v \in \R^m$ and $X \in_u S$,
\begin{align*}
|v^T (M_S - A_S) v| & = |\E[(v \cdot f(X,q))^2 - (v \cdot f(X,p))^2]| \\
& = |\E[(v \cdot f(X,q)) (v \cdot (f(X,q)-f(X,p)))]| \\
& \quad + \E[(v \cdot f(X,p)) (v \cdot (f(X,q)-f(X,p)))]| \\
& \leq \sqrt{\E[(v \cdot f(X,q))^2]\E[(v \cdot (f(X,q)-f(X,p)))^2} ]\\
& \quad +  \sqrt{\E[(v \cdot f(X,p))^2]\E[(v \cdot (f(X,q)-f(X,p)))^2]} \\
& = \sqrt{(v^T M_S v)(v^T B v)}+\sqrt{(v^T A_S v) (v^T B v)} \\
& \leq \left(\sqrt{|v^T M_S v - v^TA_S v| + \|A_S\|_2} + \|A_S\|_2\right) \sqrt{\|B\|_2} \\
& \leq  \left(\sqrt{|v^T (M_S - A_S) v|} + 2 + O(\eps)\right) \sqrt{\|B\|_2} \; .
\end{align*}

Now if $|v^T (M_S - A_S) v| \leq 4+O(\eps)$, then $|v^T (M_S - A_S) v| \leq (4 + O(\eps)) \sqrt{\|B\|_2}$ and if $||v^T (M_S - A_S) v|| \geq 4+O(\eps)$, then $|v^T (M_S - A_S) v| \leq 2 \sqrt{|v^T (M_S - A_S) v|} \sqrt{\|B\|_2}$ and so $|v^T (M_S - A_S) v| \leq 4 \|B\|_2$. Either way, we have 
$|v^T (M_S - A_S) v| \leq O(\max\{\sqrt{\|B\|_2},\|B\|_2 \})$. This holds for all $v$ and so
\[
\|M_S-A_S\|_2 \leq O(\max\{\sqrt{\|B\|_2},\|B\|_2 \} ) \; .
\]

Now consider an entry of $B$, $B_{k,\ell}=\E[(f(X,q)-f(X,p))_k (f(X,q)-f(X,p))_\ell] $. 
For any $x \in \{0, 1\}^d$ with $x \notin \Pi_k$, $F(x,p)_k-p_k=F(x,q)_k-q_k=0$.
For any $x \in \{0, 1\}^d$ with $x \in \Pi_k$, $F(x,p)_k=F(x,q)_k=x_k$ and so $(f(x,q)-f(x,p))_k=p_k-q_k$.
 Thus if both parental configurations $\Pi_k$ and $\Pi_\ell$ are true for $x$ then $(f(x,q)-f(x,p))_k (f(x,q)-f(x,p))_\ell$ is $(q-p)_k (q-p)_\ell$ and otherwise it is $0$. Thus we have
\begin{align*}
|B_{k,\ell}| & = |\E[(f(X,q)-f(X,p))_k (f(X,q)-f(X,p))_\ell]| \\
&  = \Pr_S[\Pi_k \wedge \Pi_\ell] \cdot |(q-p)_k (q-p)_\ell| \\
& \leq \min\{\Pr_S[\Pi_k], \Pr_S[\Pi_\ell]\} \cdot |(q-p)_k (q-p)_\ell| \\
& \leq \left(\sqrt{\Pr_S[\Pi_k]}(q-p)_k\right) \cdot \left(\sqrt{\Pr_S[\Pi_\ell]}(q-p)_\ell\right) \; .
\end{align*}
Now we can bound the spectral norm of $B$ in terms of its Frobenius norm:
\begin{align*}
\|B\|_2^2 & \leq \|B\|_F^2 = \sum_{k,\ell} B_{k,\ell}^2
& \leq \sum_{k,\ell}  \left(\Pr_S[\Pi_k](q-p)_k^2\right) \cdot \left(\Pr_S[\Pi_\ell](q-p)_\ell^2\right)
  \leq \left(\sum_k \Pr_S[\Pi_k](q-p)_k^2  \right)^2 \; .
\end{align*}
Combining this with the bound on $\|M_S-A_S\|_2$ above, we obtain
\[
\|M_S-A_S\|_2 \leq O(\max\{\sqrt{\sum_k \Pr_S[\Pi_k](q_k-p_k)^2},\sum_k \Pr_S[\Pi_k](q_k-p_k)^2 \} )  \; .
\]
For the diagonal entries of $M_S$ and $A_S$, we have
\begin{align*}
\|M_{S,\diag} - A_{S,\diag}\|_2 & = \max_k |M_S - A_S|_{k,k} \\
& = \max_k |\E_{X \in_u S}[f(X,q)_k^2 - f(X,p)_k^2]| \\
& = \max_k \Pr_S[\Pi_k]| \wt p_k ((1-q_k)^2-(1-p_k)^2) + (1-\wt p_k)(q_k^2-p_k^2)| \\
&= \max_k \Pr_S[\Pi_k]| 2\wt p_k (p_k-q_k) + (q_k^2-p_k^2)| \\
& = \max_k \Pr_S[\Pi_k]|2\wt p_k- p_k - q_k| |p_k-q_k| \\
& \leq \max_k 2  \Pr_S[\Pi_k] |p_k-q_k|  \\
& \leq \max_k 2  \sqrt{\Pr_S[\Pi_k]} |p_k-q_k|
  \leq 2 \sqrt{ \sum_k \Pr_S[\Pi_k] (p_k-q_k)^2} \; .
\end{align*}
Finally, we can put all this together, obtaining
\begin{align*}
\|M_{S,0}\|_2 & \leq \|M_{S}-A_{S}\|_2 + \|M_{S,\diag}-A_{S,\diag}\|_2 + \|A_{S,0}\|_2 \\
& = O(\max\{\sqrt{\sum_k \Pr_S[\Pi_k](q_k-p_k)^2}, \sum_k \Pr_S[\Pi_k](q_k-p_k)^2 \}) + O(\eps) \; . \qedhere
\end{align*}
\end{proof}

Next, we wish to bound the contribution to $M$
coming from the subtractive error. We show that this is small
due to concentration bounds on $P$ and hence on $S$. 
The idea is that for any unit vector $v$, we have tail bounds for the random variable 
$v \cdot (F(X,q)-q)$ and, since $L$ is a subset of $S$,
$L$ can at worst consist of a small fraction of the tail of this distribution. 
Then we can show that
  $\E_{X \in_u L}[(v\cdot (F(L,q)-q))^2]$ cannot be too large.

\vspace{\topsep}
\noindent {\bf Lemma \ref{lem:ML-bound}.}
{\em
$w_L \|M_L\|_2 \leq O(\eps\log(1/\eps)+\eps\|p-q\|_2^2)$.
}

\begin{proof}
Since $L \subset S$, for any event $A$, we have that $|L| \Pr_L[A] \leq |S|\Pr_S[A]$. Note that for any $x$, since $((F(x,q)-q)-(F(x,p)-p))_i$ is either $0$ or $p_i-q_i$ for any $i$, thus $\|(F(X,q)-q)-(F(X,p)-p)\|_2 \leq \|p-q\|_2$.
Since $S$ is $\eps$-good for $P$, by Lemma~\ref{lem:good-set} (iii), we have
\[
|L| \Pr_{X \in_u L}\left[|v \cdot (F(X,q)-q)| \geq T + \|p-q\|_2\right] \leq |S|(3\exp(-T^2/2)+\eps/(T^2 \ln d)) \;.
\]
Also not that $\Pr_{X \in_u L}[|v \cdot (F(X,q)-q)| > \sqrt{d}]=0$ since $\|F(X,q)-q\|_2 \leq \sqrt{d}$.
 By definition,
$\|M_L\|_2$ is the maximum over unit vectors $v$ of $v^T M_L v$.
For any unit vector $v$, we have~\footnote{We write $f(x) \ll g(x)$ for $f(x) = O(g(x)).$}
\begin{align*}
& |L| v^T M_L v = |L| \cdot \E_{X \in_u L}\left[ (v\cdot (F(X,q)-q))^2 \right] \\
& =  2 |L| \int_{0}^{\sqrt{d}} \Pr_{X \in_u L} \left[ |v\cdot (F(X,q)-q)| \geq T \right] T dT \\
& \ll \int_{0}^{2 \|p-q\|_2 + 2 \sqrt{\ln(|S|/|L|)}} |L| T dT
  + \int_{\|p-q\|_2 + 2 \sqrt{\ln(|S|/|L|)}}^{\sqrt{d} - \|p-q\|_2} |S|\exp\left(-T^2/2 \right) (T + \|p-q\|_2) dT \\
& \quad + \int_{\|p-q\|_2 + 2 \sqrt{\ln(|S|/|L|)}}^{\sqrt{d} - \|p-q\|_2} |S|\eps (T + \|p-q\|_2)/(T^2 \log d)  dT\\
& \ll \int_{0}^{2 \|p-q\|_2 + 2 \sqrt{\ln(|S|/|L|)}} |L| T dT
  + \int_{2 \sqrt{\ln(|S|/|L|)}}^{\infty} |S|\exp\left(-T^2/2 \right) T dT
  + \int_{1}^{\sqrt{d}} |S|\eps/(T \log d)  dT\\
& \ll |L| \left(\|p-q\|_2^2 + \log(|S|/|L|)\right) + |L| + \eps|S|\\
& \ll \eps\log(1/\eps)|S'|+ \eps |S'|\|p-q\|_2^2 \; .
\end{align*}
The last inequality uses $|L |\leq 2\eps|S'|$ and $|S| \leq (1+2\eps)|S'|$.
\end{proof}

Finally, combining the above results,
since $M_{S}$ and $M_L$ have small contribution to the spectral norm of $M$ when $\|p-q\|_2$ is small,
most of it must come from $M_E$.

\vspace{\topsep}
\noindent {\bf Lemma \ref{lem:MApprox}.}
{\em
$\|M - w_E M_E\|_2 \leq O(\eps \log(1/\eps) + \sqrt{\sum_k \Pr_{S'}[\Pi_k] (p_k-q_k)^2} + \sum_k \Pr_{S'}[\Pi_k] (p_k-q_k)^2)$.
}

\begin{proof}
Note that $|S'|M=|S| M_{S,0} + |E| M_{E,0} - |L| M_{L,0}$.

Note that each entry of any of these matrices has absolute value at most one since $|F(x,q)-q|_k \leq 1$ for all $x \in \{0,1\}^d$ and $k \in [m]$. Thus we have
\[
\|M_{L,0}\| \leq \|M_L\|_2 + \|M_L-M_{L,0}\|_2 = \|M_L\|_2 + \max_k |(M_L)_{k,k}| \leq \|M_L\|_2 + 1
\]
and similarly,
\[ \|M_E-M_{E,0}\|_2 \le \max_k |(M_E)_{k,k}| \le 1 \;. \]

By the triangle inequality, Lemmas~\ref{lem:MP-bound} and \ref{lem:ML-bound}, and the assumption that $|E| + |L| \le 2\eps |S'|$,
\begin{align*}
&  \||S'|M - |E| M_E\|_2 \leq |S|\|M_{S,0}\|_2  + |L| \|M_{L,0}\|_2 +|E|\|M_E-M_{E,0}\|_2 \\
& \leq |S| \cdot  O\left(\sqrt{\sum_k \Pr_S[\Pi_k] (p_k-q_k)^2} + \sum_k \Pr_S[\Pi_k] (p_k-q_k)^2 + \eps\log(1/\eps)+\eps\|p-q\|_2^2\right) \;,
\end{align*}

Using Lemma~\ref{lem:alpha-good}, we obtain that
$\eps \|p-q\|_2^2  \leq \sum_k \Pr_{S'}[\Pi_k] (p_k-q_k)^2$ and
$\sum_k \Pr_S[\Pi_k] (p_k-q_k)^2 \leq \sum_k (\Pr_{S'}[\Pi_k]+2\eps) (p_k-q_k)^2 = O(\sum_k \Pr_{S'}[\Pi_k] (p_k-q_k)^2)$.
\end{proof}

\subsection{Omitted Proofs from Section~\ref{ssec:small-norm}: The Case of Small Spectral Norm} \label{app:small-norm}

In this section, we will prove that if $\|M\|_2 = O(\eps \log(1/\eps)/\alpha)$, then we can output the empirical conditional means $q$.

We first show that the contributions that $L$ and $E$ make to $\E_{X \in_u S'}{[F(X,q)-q)}]$ can be bounded in terms of the spectral norms of $M_L$ and $M_E$.

\vspace{\topsep}
\noindent {\bf Lemma \ref{lem:wrong-mean-covar}.}
{\em
$\left\|\E_{X \in_u L}[F(X,q)] - q \right\|_2 \leq \sqrt{\|M_L\|_2}$ and $\|\E_{X \in_u E}[F(X,q)] - q\|_2 \leq \sqrt{\|M_E\|_2}$.
}

\begin{proof}
Let $Y$ be any random variable supported on $\R^m$ and $y \in \R^m$.
Then  we have
\[
\E[(Y-y)(Y-y)^T]=\E\left[(Y-\E[Y])(Y-\E[Y])^T\right] + (\E[Y]-y)(\E[Y]-y)^T \;.
\]
Since both terms are positive semidefinite,
we have
\[
\| \E{\left[(Y-y)(Y-y)^T\right]} \|_2 \geq \|(\E[Y]-y)(\E[Y]-y)^T\|_2 = \|\E[Y]-y\|_2^2 \;.
\]
Applying this with $y=q$ and $Y=F(X,q)$ for $X \in_u L$ (or $X \in_u E$) completes the proof.
\end{proof}

Combining with the results about these norms in the previous section, Lemma~\ref{lem:wrong-mean-covar} implies that if $\|M\|_2$ is small, then $q = \E_{X \in_u S'}[F(X,q)]$ is close to $\E_{X \in_u S}[F(X,q)]$, which is then necessarily close to $\E_{X \sim P}[F(X,p)] = p$.
The following lemma states that the mean of $(F(X,q)-q)$ under the good samples is close to $(p-q)$ scaled by the probabilities of parental configurations under $S'$.

\vspace{\topsep}
\noindent {\bf Lemma \ref{lem:mean-transformed}.}
{\em
Let $z \in \R^m$ be the vector with $z_k=\Pr_{S'}[\Pi_k] (p_k-q_k)$.  Then,
$\| \E_{X \in_u S}[F(X,q)-q] - z\|_2 \leq O(\eps (1+\|p-q\|_2)).$
}
\begin{proof}
When $\Pi_k$ does not occur, $F(X,q)_k=q_k$.
Thus, we can write:
\begin{align*}
\E_{X \in_u S}[F(X,q)_k-q_k] & = \Pr_S[\Pi_k] \E_{X \in_u S}\left[F(X,q)_k-q_k \mid \Pi_k \right] \\
& = \Pr_S[\Pi_k] (\wt p_k-q_k) =  \Pr_S[\Pi_k] (p_k-q_k) + \Pr_S[\Pi_k] (\wt p_k-p_k) \;,
\end{align*}
Since $S$ is $\eps$-good, by Lemma~\ref{lem:good-set} (ii), we know that $\sum_k \Pr_S[\Pi_k]^2 (\wt p_k-p_k)^2 \leq \eps^2$ and so if $z'$ is the vector with $z'_k=\Pr_{S}[\Pi_k] (p_k-q_k)$, then  $\|\E_{X \in_u S}[F(X,q)-q] - z\|_2 \leq \eps$.

By Lemma~\ref{lem:alpha-good}, $|\Pr_{S'}[\Pi_k] - \Pr_S[\Pi_k]| \leq 2\eps$, so we have $|\Pr_{S'}[\Pi_k] (p_k-q_k) - \Pr_{S}[\Pi_k] (p_k-q_k)| \leq 2\eps|p_k-q_k|$, i.e., $|z_k - z'_k| \le 2\eps |p_k - q_k|$, and thus $\|z-z'\|_2 \leq 2\eps \|p-q\|_2$.
\end{proof}

Note that $z = \left(\Pr_{S'}[\Pi_k] (p_k-q_k)\right)_{k=1}^m$ is closely related to the total variation distance between $P$ and $Q$ (see Lemma~\ref{cor:bal-min}).
We can write $(\E_{X \in_u S'}[F(X,q)]-q)$ in terms of this expectation under $S, E$, and $L$ whose distance from $q$ can be upper bounded using the previous lemmas.
Using Lemmas~\ref{lem:ML-bound},~\ref{lem:MApprox},~\ref{lem:wrong-mean-covar},~and~\ref{lem:mean-transformed}, we can bound $\|z\|_2$ in terms of $\|M\|_2$.

\vspace{\topsep}
\noindent {\bf Lemma \ref{lem:mean-dist-from-norm}.}
{\em
$\sqrt{\sum_k \Pr_{S'}[\Pi_k]^2 (p_k-q_k)^2} \leq 2\sqrt{\eps \|M\|_2} + O(\eps \sqrt{\log(1/\eps)+1/\alpha})$.
}

\begin{proof}
For notational simplicity, let $D \in \R^{m \times m}$ be a diagonal matrix with $D_{k,k}=\diag(\sqrt{\Pr_{S'}[\Pi_k]})$.
We have
\[
\sqrt{\sum_k \Pr_{S'}[\Pi_k] (p_k-q_k)^2} = \|D(p-q)\|_2 \; .
\]

Let $\mu^{S'}$, $\mu^{S}$, $\mu^L$ and $\mu^E$ be $\E[F(X,q)]$ for $X$ taken uniformly from $S'$, $S$, $L$ or $E$ respectively.
We have the identity
\[
|S'|\mu^{S'} = |S| \mu^S - |L| \mu^L + |E| \mu^E \;.
\]
Or equivalently,
\[
|S'|(\mu^S - \mu^{S'}) = (|S'| - |S|) \mu^S + |L| \mu^L - |E| \mu^E \;.
\]

Note that $\mu^{S'}=q$.
By Lemma~\ref{lem:mean-transformed}, $\|(\mu^{S} - q) -  (D^2(p-q))\|_2 \leq O(\eps(1+\|p-q\|_2))$.
Recall that $w_L=|L|/|S'|$ and $w_E=|E|/|S'|$.
By the triangle inequality,
\begin{align*}
& \|D^2(p-q)\|_2 \le \|\mu^{S}-q\|_2 +  O(\eps(1+\|p-q\|_2) \\
& = \|w_L(\mu^L-q) - w_E(\mu^E-q)+(1-w_S)(\mu^S-q)\|_2 + O(\eps+ \eps\|p-q\|_2))\\
& \leq w_L\|\mu^L-q\|_2 + w_E\|\mu^E-q\|_2 + 2\eps \|(\mu^{S} - q)\|_2 +  O(\eps+ \eps\|p-q\|_2)\\
& \leq w_L\|\mu^L-q\|_2 + w_E\|\mu^E-q\|_2 + O\left(\eps\|D^2(p-q)\|_2\right) + O(\eps+ \eps\|p-q\|_2)) \\
& \leq w_L \sqrt{\|M_L\|_2} + w_E \sqrt{\|M_E\|_2} + O\left(\eps\|D^2(p-q)\|_2\right) + O(\eps+ \eps\|p-q\|_2))\\
& \leq O\left(\eps \sqrt{\log(1/\eps)}\right) + \sqrt{2\eps \|M\|_2} + O\left(\sqrt{\eps}\|D(p-q)\|_2 + \sqrt{\eps\|D(p-q)\|_2}\right) \\
& \leq O\left(\eps \sqrt{\log(1/\eps)}\right) + \sqrt{2\eps \|M\|_2} + O\left(\sqrt{\eps/\alpha}\|D^2(p-q)\|_2 + \sqrt{\eps\|D^2(p-q)\|_2/\sqrt{\alpha}}\right) \\
& \leq O\left(\eps \sqrt{\log(1/\eps)}\right) + (3/2)\sqrt{\eps \|M\|_2} + \|D^2(p-q)\|_2/8 + O\left(\sqrt{\eps\|D^2(p-q)\|_2/\sqrt{\alpha}}\right) \;,
\end{align*}
where we used the assumption that $\alpha/\eps$ is at least a sufficiently large constant and that $\|D^{-1}\|_2=1/\sqrt{\alpha}$.
When this last term is smaller than $\|D^2(p-q)\|_2/8$, rearranging the inequality gives that
\[
\|D^2(p-q)\|_2 \leq 2\sqrt{\eps \|M\|_2} + O(\eps \sqrt{\log(1/\eps)}) \; .
\]
Otherwise, we have $\|D^2(p-q)\|_2= O\left(\sqrt{\eps\|D^2(p-q)\|_2/\sqrt{\alpha}}\right)$, and so $\|D^2(p-q)\|_2 = O(\eps/\sqrt{\alpha})$. In either case, we obtain
\[
\|D^2(p-q)\|_2 \leq 2\sqrt{\eps \|M\|_2} + O(\eps \sqrt{\log(1/\eps) + 1/\alpha}) \; . \qedhere
\]
\end{proof}

Lemma~\ref{lem:mean-dist-from-norm} implies that, if $\|M\|_2$ is small then so is $\sqrt{\sum_k \Pr_{S'}[\Pi_k]^2 (p_k-q_k)^2}$.
We can then use it to show that $\sqrt{\sum_k \Pr_P[\Pi_k] (p_k-q_k)^2}$ is small.
We can do so by losing a factor of $1/\sqrt{\alpha}$ to remove the square on $\Pr_{S'}[\Pi_k]$, and showing that $\min_k \Pr_S'[\Pi_k] = \Theta(\min_k \Pr_P[\Pi_k])$ when it is at least a large multiple of $\eps$.
Finally, if $\sqrt{\sum_k \Pr_P[\Pi_k] (p_k-q_k)^2}$ is small, Lemma~\ref{cor:bal-min} tells us that $\dtv(P,Q)$ is small.
This completes the proof of the first case of Proposition~\ref{prop:bal}.

\vspace{\topsep}
\noindent {\bf Corollary \ref{cor:correct}}\;\;(Part {\em (i)} of Proposition~\ref{prop:bal}){\bf .}
{\em
If $\|M\|_2 \leq O(\eps \log(1/\eps)/\alpha)$, then

$ \qquad \dtv(P,Q)=O(\eps\sqrt{\log(1/\eps)}/(c \min_k \Pr_P[\Pi_k])).$
}
\begin{proof}
Recall that $\alpha = \min_k |\Pr_{S'}[\Pi_k]$.
By Lemma~\ref{lem:alpha-good}, $|\Pr_{S'}[\Pi_k]-\Pr_{S}[\Pi_k]| \leq 2\eps$ for any $k$.
Since $S$ is $\eps$-good, $|\Pr_{P}[\Pi_k]-\Pr_{S}[\Pi_k]| \leq \eps$.
Combining these, we obtain $|\alpha-\min_k \Pr_P[\Pi_k]| \leq 3\eps$.
By assumption $\min_k \Pr_P[\Pi_k] \geq 4\eps$ , so we have $\alpha = \Theta(\min_k \Pr_P[\Pi_k])$.

Therefore,
\begin{align*}
\sqrt{\sum_k \Pr_P[\Pi_k] (p_k-q_k)^2}
&\leq  \sqrt{\sum_k \Pr_P[\Pi_k]^2 (p_k-q_k)^2}/\sqrt{\min_k \Pr_P[\Pi_k]} \\
&\leq \left(2\sqrt{\eps \|M\|_2} + O(\eps \sqrt{\log(1/\eps)+1/\alpha})\right)/\sqrt{\min_k \Pr_P[\Pi_k]} \tag*{(by Lemma \ref{lem:mean-dist-from-norm})} \\
&\leq O\left(\eps \sqrt{\log(1/\eps)}/\min_k \Pr_P[\Pi_k]\right) \;.
\end{align*}
Now we can apply Lemma~\ref{cor:bal-min} to get $\dtv(P,Q) \le O\left(\eps \sqrt{\log(1/\eps)}/(c\min_k \Pr_P[\Pi_k])\right)$.
\end{proof}

\subsection{Omitted Proofs from Section~\ref{ssec:big-norm}: The Case of Large Spectral Norm} \label{app:big-norm}
Now we consider the case when $\|M\|_2 \geq C\eps \ln(1/\eps)/\alpha$ for some sufficiently large constant $C > 0$.
We begin by showing that $p$ and $q$ are not too far apart from each other.
The bound given by Lemma \ref{lem:mean-dist-from-norm} is 
now dominated by the $\|M\|_2$ term. 
Lower bounding the $\Pr_{S'}[\Pi_k]$ by $\alpha$ gives the following claim.

\vspace{\topsep}
\noindent {\bf Claim \ref{clm:delta-distance}.}
{\em
$\|p-q\|_2 \leq \delta := 3 \sqrt{\eps\|M\|_2}/\alpha$.
}

\begin{proof}
By Lemma \ref{lem:mean-dist-from-norm}, we have that
\[\sqrt{\sum_k \Pr_{S'}[\Pi_k]^2 (p_k-q_k)^2} \leq 2\sqrt{\eps \|M\|_2} + O(\eps \sqrt{\log(1/\eps) + 1/\alpha}) \; .\]
For sufficiently large $C$, this last term is smaller than
$\eps \sqrt{C \ln(1/\eps)/\alpha} \le \frac{1}{2} \sqrt{\eps \|M\|_2}$.
Then we have $\sqrt{\sum_k \Pr_{S'}[\Pi_k]^2 (p_k-q_k)^2} \leq (5/2)\sqrt{\eps \|M\|_2}$.
Recall that $\alpha = \min_k \Pr_{S'}[\Pi_k]$, so
\[
\alpha \|p-q\|_2 = \sqrt{\sum_k \alpha^2 (p_k-q_k)^2} \le \sqrt{\sum_k \Pr_{S'}[\Pi_k]^2(p_k-q_k)^2} \leq (5/2)\sqrt{\eps \|M\|_2} \; . \qedhere
\]
\end{proof}

Recall that $v^*$ is the largest eigenvector of $M$.
We project all the points $F(X,q)$ onto the direction of $v^*$.
Next we show that most of the variance of $(v^*\cdot (F(X,q) - q))$ comes from $E$.

\vspace{\topsep}
\noindent {\bf Claim \ref{clm:norm-mostly-E}.}
{\em 
$v^{\ast T} (w_E M_E) v^{\ast} \geq \frac{1}{2} v^{\ast T} M v^{\ast}$.
}

\begin{proof}
By Lemma \ref{lem:MApprox} and Claim \ref{clm:delta-distance}, we deduce
 \begin{align*}
\|M - w_E M_E\|_2 
&\leq O\left(\eps \log(1/\eps) + \sum_k \Pr_{S'}[\Pi_k] (p_k-q_k)^2 + \sqrt{\sum_k \Pr_{S'}[\Pi_k] (p_k-q_k)^2} \right) \\
&\leq O\left(\eps \log(1/\eps) + \eps\|M\|_2/\alpha + \sqrt{\eps\|M\|_2/\alpha} \right) \;.
\end{align*}
By assumption $\min_k \Pr_P[\Pi_k] \geq C'\eps$ for sufficiently large $C'$, so we can assume $\eps/\alpha \leq 1/6$. 
For large enough $C$, $\|M\|_2 \geq C\eps \ln(1/\eps)/\alpha \geq 36\eps/\alpha$,
  and hence the third term $\sqrt{\eps\|M\|_2/\alpha} \leq \|M\|_2/6$.
Again for large enough $C$, the first term is upper bounded by $\|M\|_2/6$.
Thus, we obtain $2 \|M - w_E M_E\|_2 \leq \|M\|_2 = v^{\ast T} M v^{\ast}$ as required.
\end{proof}

Claim~\ref{clm:norm-mostly-E} implies that the tails of $E$ are reasonably thick.
In particular, the next lemma shows that is guaranteed to find some valid threshold $T>0$ satisfying the desired property in Step~\ref{step:findT} of Algorithm~\ref{alg:filter}, otherwise by integrating the tail bound, we can show that $v^{\ast T} M_E v^{\ast}$ would be small.

\vspace{\topsep}
\noindent {\bf Lemma \ref{lem:exists-T-dist}.}
{\em 
There exists a $T \geq 0$ such that

$
\qquad \Pr_{X \in_u S'}[|v\cdot (F(X,q)-q)|>T+\delta] > 7\exp(-T^2/2)+3\eps/(T^2 \ln d).
$
}

\begin{proof}
Suppose for the sake of contradiction that this does not hold.
Since $E \subset S'$, for all events $A$, it holds that $|E| \cdot \Pr_E[A] \leq |S'| \cdot \Pr_{S'}[A]$.
Thus, we have
\[
w_E \Pr_{X \in_u E}[|v^\ast \cdot (F(X,q) - q)| \geq T + \delta] \leq 7\exp(-T^2/2)+3\eps/(T^2 \ln d) \;.
\]
Note that for any $x \in \{0,1\}^d$,
we have that
\[
|v^\ast \cdot (F(x,q)-q)| \leq \|F(x,q)-q\|_2 \leq \sqrt{d} \;,
\]
since $F(x,q)$ and $q$ differ on at most $d$ coordinates. We have the following sequence of inequalities:
\begin{align*}
& \|M\|_2 \ll w_E v^{\ast T} M_E v^{\ast} \\
& = 2 w_E \int_0^{\sqrt{d}} \Pr_{X \in_u E}\left[|v^\ast \cdot (F(X,q)-q)| \geq T\right] T dT \\
& \leq 2 w_E \int_0^{2 \delta + 2 \sqrt{\ln(1/w_E)}} T dT + \int_{\delta + 2 \sqrt{\ln(1/w_E)}}^{\sqrt{d} - \delta} \left(7 \exp(-T^2/2)+3\eps/(T^2 \ln d) \right) (T + \delta) dT \\
& \ll w_E \int_0^{2 \delta + 2 \sqrt{\ln(1/w_E)}} T dT + \int_{\delta + 2 \sqrt{\ln(1/w_E)}}^{\infty} \exp(-T^2/2) T dT + \int_{1}^{\sqrt{d}} \eps / (T \log d) \\
& \ll w_E \delta^2 + w_E \log(1/w_E) + \eps \\
& \ll \eps \delta^2 + \eps \log(1/\eps) \\
& \ll (\eps^2/\alpha^2)\|M\|_2 + \alpha \|M\|_2/C \\
& \ll \|M\|_2/(C'-3)^2 + \|M\|_2/C \; .
\end{align*}
For sufficiently large $C'$ and $C$, this gives the desired contradiction.
\end{proof}

Finally, we show that the set of samples $S''$ we return after the filter 
is better than $S'$ in terms of $|L|+|E|$.
This completes the proof of the second case of Proposition~\ref{prop:bal}.

\vspace{\topsep}
\noindent {\bf Claim \ref{clm:calc}.}\;\;(Part {\em (ii)} of Proposition~\ref{prop:bal}){\bf .}
{\em 
If we write $S''=S \cup E' \setminus L'$, then $|E'|+|L'| < |E|+|L|$ and $|S''| \le (1-\frac{\eps}{d\ln d})|S'|$.
}

\begin{proof}
Note that $(F(x,q)-q)$ is $d$-sparse and has $\| F(x,q)-q \|_\infty \le 1$, so we have $|v \cdot (F(x,q)-q)| \le \| F(x,q)-q \|_2 \le \sqrt{d}$.
By Lemma~\ref{lem:exists-T-dist}, when $\|M\|_2 \geq C\eps \ln(1/\eps)/\alpha$ for some sufficiently large constant $C > 0$, Step~\ref{step:findT} of Algorithm~\ref{alg:filter} is guaranteed to find a threshold $0 < T \le \sqrt{d}$ such that
\[
\Pr_{X \in_u S'}[|v\cdot (F(X,q)-q)|>T+\delta] > 7\exp(-T^2/2)+3\eps/(T^2 \ln d) \;.
\]
Thus, we have
\[
|S'| - |S''| > (7\exp(-T^2/2)+3\eps/(T^2 \ln d))|S'| \;.
\]
In particular, we can show that the number of remaining samples reduces by a factor of $(1-\eps/(d \ln d))$:
\[
|S'| - |S''| > 3\eps/(T^2 \ln d) \cdot |S'| \ge (\eps/(d \ln d) \cdot |S'| \;.
\]

Since $S$ is $\eps$-good, by Lemma \ref{lem:good-set} (iii),
\[
\Pr_{X \in_u S}[|v \cdot (F(X, p) - p)| \geq T] \leq 3 \exp(-T^2/2) + \eps/(T^2 \ln d) \; .
\]
Using Claim \ref{clm:delta-distance}, we have that for all $x \in \{0,1\}^d$, $\|(F(x,q)-q)-(F(x,p)-p)\|_2 \leq \|p-q\|_2 \leq \delta$.
Therefore,
\[
\Pr_{X \in_u S}[|v \cdot (F(X, q) - q)| \geq T + \delta] \leq 3 \exp(-T^2/2) + \eps/(T^2 \ln d) \; .
\]
Since all the filtered samples $L' \setminus L$ are in $S$, we have
\[
|L'| - |L| \leq (3 \exp(-T^2/2) + \eps/(T^2 \ln d))|S| \; .
\]
Thus $|S'|-|S''| \geq \frac{7}{3}(1-2\eps)(|L'| - |L|)$.
Since $|S'|-|S''|=|E|-|L| - |E'|+|L'|$,
\[
|E|+|L| - |E'|-|L'| = (|S'|-|S''|) - 2(|L'|-|L''|) \geq (1/7 - O(\eps)) (|S'| - |S''|) \; .
\]
Because $S'' \subset S'$, we conclude that $|E'|+|L'| < |E|+|L|$.
\end{proof}

\section{Omitted Details from Section \ref{sec:experiment}}
\label{apx:experiment}

\subsection{Details of the Experiments}
\label{apx:exp-noise}
In this section, we give a detailed description of the graph structures 
and noise distributions we used in our experimental evaluation.

In our experiments, when there is randomness in the dependency graphs 
of the ground-truth or in the noisy Bayesian networks, 
we repeat the experiment ten times   and report the average error.

\paragraph{Synthetic Experiments with Tree Bayesian Networks.}
In the first experiment, the ground-truth Bayesian network $P$ is generated as follows:
We first generate a random dependence tree for $P$.
We label the $d$ nodes $\{1, \ldots, d\}$.
Node $1$ has no parents, and every node $i > 1$ has one parent drawn uniformly from $\{1, \ldots, i-1\}$.
The size of the conditional probability table of $P$ is $2d-1$ 
(one parameter for the first node and two parameters for all other nodes).
We then draw these $(2d-1)$ conditional probabilities independently and uniformly from $[0, \frac{1}{4}] \cup [\frac{3}{4}, 1]$.

The noise distribution is a binary product distribution 
with mean drawn independently and uniformly from $[0, 1]$ for each coordinate.

\paragraph{Synthetic Experiments with General Bayesian Networks.}
In the second experiment, we generate the ground-truth network $P$ as follows:
We start with an empty dependency graph with $d = 50$ nodes.
We label the $d$ nodes $\{1, \ldots, d\}$ and require that parent nodes must have smaller index.
We continue to try to increase the in-degree of a random node until the number of parameters $m = \sum_{i=1}^d 2^{|\text{Parent}(i)|}$ exceeds the target $m \in [100, 1000]$.
Then for each $i$, we draw the $|\text{Parents(i)}|$ nodes uniformly from the set $\{1, \ldots, i-1\}$ to be the parents of variable $i$.

The noise distribution is a tree-structured Bayes net generated in the same way as we generate the ground-truth network in the first experiment.
The conditional probabilities of both $P$ and the noise distribution are drawn independently and uniformly from $[0, \frac{1}{4}] \cup [\frac{3}{4}, 1]$.

\paragraph{Semi-Synthetic Experiments with ALARM.}
In the third experiment, the ground-truth network is a binary-valued 
Bayes net which is equivalent to the ALARM network.
See Section~\ref{apx:exp-multi-binary} for a detailed description of the conversion process.
Specifically, it has $d = 61$ nodes and $m = 820$ parameters.

The noise distribution is a Bayes net generated using the same process as we create the ground-truth Bayes net in the second experiment.
We start with an empty graph with $d = 61$ nodes and add edges until the number of parameters is roughly $m = 820$.
The conditional probabilities of the noise distribution are again drawn from $[0, \frac{1}{4}] \cup [\frac{3}{4}, 1]$.

\subsection{Estimating the Total Variation Distances between Two Bayesian Networks}
\label{apx:exp-dtv}
Given two $d$-dimensional Bayesian networks $P$ and $Q$ 
(explicitly with their dependency graphs and conditional probability tables), we want to compute the total variation distance between them.
Since there is no closed-form formula, we use sampling to estimate $\dtv(P, Q)$.
By definition,
\[
\dtv(P, Q) = \sum_{x \in \{0,1\}^d, P(x) > Q(x)} \left(P(x) - Q(x)\right) \;.
\]
Let $A = \{x \in \{0,1\}^d: P(x) > Q(x) \}$. We have $\dtv(P, Q) = P(A) - Q(A)$ where $P(A) = \sum_{x \in A} P(x)$ and $Q(A) = \sum_{x \in A} Q(x)$.
We can draw samples from $P$ to estimate $P(A)$, 
since for a fixed $x \in \{0, 1\}^d$, we can efficiently test whether $x \in A$ or not 
by computing the log-likelihood of $x$ under $P$ and $Q$.
Similarly, we can estimate $Q(A)$ by drawing samples from $Q$.

In all of our evaluations, we take $N = 10^6$ samples from $P$ to estimate $P(A)$ (and similarly for $Q(A)$).
By Hoeffding's inequality, the probability that our estimate of $P(A)$ is off by more than $\eps$ is at most $2 \exp(-2 N \eps^2)$.
For example, with probability $0.99$, we can estimate both $P(A)$ and $Q(A)$ within additive error $0.2\%$, 
  which gives an estimate for $\dtv(P, Q)$ within additive error $0.4\%$.

\subsection{Reduction to Binary-Valued Bayesian Networks}
\label{apx:exp-multi-binary}

The results in this paper can be easily extended to multi-valued Bayesian networks.
We can represent a $d$-dimensional degree-$f$ Bayesian network 
over alphabet $\Sigma$ by an equivalent binary (i.e., alphabet of size $2$) Bayesian network of dimension 
$d \ceil{\log_2(|\Sigma|)}$ and degree $(f+1)\ceil{\log_2(|\Sigma|)}$.
Such a reduction can be found in \cite{CanonneDKS17} and we give a high-level description here.
Without loss of generality we can assume $|\Sigma| = 2^b$.
We will split each variable into $b$ bits, with each of the $2^b$ possibilities denoting a single letter in $\Sigma$.
Each new bit will potentially depend on other bits of the same variable,
as well as bits of the parent variables.
This operation preserves balancedness when $|\Sigma| = 2^b$,
and if $|\Sigma|$ is not a power of $2$ we need to first carefully 
pad the alphabet by splitting some letters in $\Sigma$ into two letters.

Our experiments for the ALARM network use this reduction.
ALARM has an alphabet of size $4$. The original dependency graph of ALARM has $37$ nodes, maximum in-degree $4$, and $509$ parameters;
and after the transformation, we get a binary-valued the network with $61$ nodes, maximum in-degree $7$, and $820$ parameters. 
}{}

\end{document}